\documentclass{llncs}

\usepackage{booktabs}   

\usepackage{ulem} 

\usepackage{inconsolata} 

\usepackage{listings}
\usepackage{url}

\usepackage{amsmath}
\usepackage{amssymb}
\usepackage{stmaryrd}

\usepackage{tabularx}
\usepackage{mathpartir}
\usepackage{ifthen}
\usepackage{balance}
\usepackage{listings-golang} 
\usepackage{color}
\definecolor{GrayBgColor}{rgb}{0.9, 0.9, 0.9}
\definecolor{GrayFgColor}{rgb}{0.4, 0.4, 0.4}
\definecolor{StringColor}{rgb}{0.0, 0.38039, 0.141176}
\usepackage{graphicx}

\makeatother
\makeatletter
\newcommand\gokw[1]{{\bf\ttfamily{}#1}}
\newcommand\go[1]{\lstinline[language=Golang]{#1}}

\lstdefinelanguage{myhaskell}%
  {morekeywords={if,then,else,case,class,data,default,deriving,%
      hiding,if,in,infix,infixl,infixr,import,instance,let,module,%
      newtype,of,qualified,type,where,do,forall},%
   sensitive,%
   morecomment=[l]--,%
   morecomment=[n]{\{-}{-\}},%
   morestring=[b]",%
   keywordstyle=\bf\ttfamily,%
   literate={REC}{{$\texttt{RE}_{\texttt{C}}$}}3 {RC}{{$\texttt{R}_{\texttt{C}}$}}2 {cast_RE}{{$\texttt{cast}_{\texttt{RE}}$}}5 {ListC}{{$\texttt{List}_{\texttt{C}}$}}5 {FunC}{{$\texttt{Fun}_{\texttt{C}}$}}4%
  }[keywords,comments,strings]%

\renewcommand\emph[1]{{\it #1}}

\def\CC{{C\nolinebreak[4]\hspace{-.05em}\raisebox{.4ex}{\tiny\bf ++}}}

\newcommand\mcolorbox[2]{\colorbox{#1}{\ensuremath{#2}}}

\newcommand{\mathem}{\sf}
\newcommand{\kw}[1]{\normalfont\bfseries\sffamily #1}
\newcommand{\IN}{\mbox{\kw{in}}}
\newcommand{\LET}{\mbox{\kw{let}}}

\newcommand{\CASE}{\mbox{\kw{case}}}

\newcommand{\OF}{\mbox{\kw{of}}}

\newcommand{\foreach}[2]{\overline{#1}^{#2}} 
\newcommand{\foreachN}[1]{\overline{#1}}    

\newcommand{\gap}{\,\,\,\,}

\newcommand{\TYPE}{\mbox{\kw{type}}}
\newcommand{\STRUCT}{\mbox{\kw{struct}}}
\newcommand{\INTERFACE}{\mbox{\kw{interface}}}
\newcommand{\FUNC}{\mbox{\kw{func}}}
\newcommand{\RETURN}{\mbox{\kw{return}}}
\newcommand{\PACKAGE}{\mbox{\kw{package}}}
\newcommand{\MAIN}{\mbox{\mathem main}}

\newcommand\GoSynCatName[1]{\mbox{\small #1}}

\newcommand{\methodSpecificationsRel}{{\mathem methods}}
\newcommand{\methodSpecifications}[2]{\methodSpecificationsRel(#1,#2)}

\newcommand{\fgEnv}{\Delta}
\newcommand{\EmptyFgEnv}{\{\}}
\newcommand{\turnsFG}{\, \vdash_{\mathsf{FG}} \,}

\newcommand{\subtypeOf}[2]{#1 <: #2}
\newcommand\assertOfSym{\mathrel{\scalebox{0.8}{\ensuremath{\searrow}}}}
\newcommand{\assertOf}[2]{#1 \assertOfSym #2}
\newcommand{\subtypeOfSym}{<:}
\newcommand{\fgSub}[3]{#1 \turnsFG \subtypeOf{#2}{#3}}

\newcommand{\FGStructNonRec}{{FG1}}
\newcommand{\FGUniqueFields}{{FG2}} 
\newcommand{\FGUniqueMethSpec}{{FG3}} 
\newcommand{\FGUniqueReceiver}{{FG4}} 

\newcommand{\EvCtx}{{\mathcal E}}
\newcommand{\reduceSym}{\longrightarrow}
\newcommand{\reduce}[2]{#1 \reduceSym  #2}
\newcommand{\reduceN}[2]{#1 \reduceSym^* #2} 
\newcommand{\subst}[2]{#1 \mapsto #2}

\newcommand{\reduceFGk}[4]{#2 \turnsFG #3 \reduceSym^{#1} #4} 

\newcommand{\reduceTL}[3]{#1 \turnsTL #2 \reduceSym #3}
\newcommand{\reduceTLN}[3]{#1 \turnsTL #2 \reduceSym^* #3} 
\newcommand{\reduceTLk}[4]{#2 \turnsTL #3 \reduceSym^{#1} #4} 

\newcommand{\irred}{{\mathem irred}}

\newcommand{\vbName}{\Phi}

\newcommand{\vbFG}{\vbName_{\mathsf{v}}} 

\makeatletter
\newcommand{\xRightarrow}[2][]{\ext@arrow 0359\Rightarrowfill@{#1}{#2}}
\makeatother

\newcommand{\TL}{\mbox{TL}} 

\newcommand{\kT}{K}         

\newcommand\Angle[1]{\langle#1\rangle}
\newcommand{\pair}[2]{\Angle{#1,#2}}
\newcommand{\triple}[3]{\Angle{#1,#2, #3}}

\newcommand{\program}{\mathit{Prog}}

\newcommand{\clsT}{\mathit{Cls}}
\newcommand{\patT}{\mathit{Pat}}
\newcommand{\expT}{E}

\newcommand{\xT}{X}
\newcommand{\yT}{Y}

\newcommand{\xTrep}{\xT_{\mathit{rep}}}
\newcommand{\xTval}{\xT_{\mathit{val}}}
\newcommand{\vT}{V}

\newcommand{\uT}{V}                 

\newcommand{\ignore}[1]{}

\newcommand{\tcaseof}[2]{\CASE\ #1 \ \OF\ #2}

\newcommand{\tletrecin}[2]{\LET\ #1 \ \IN\ #2}  

\newcommand\turnsSomething[1]{\, \vdash_{\mathsf{#1}} \,}
\newcommand\turnsSomethingk[2]{\, \vdash^{#2}_{\mathsf{#1}} \,}

\newcommand{\turnsTL}{\turnsSomething{TL}}

\newcommand{\vbTL}{\vbName_{\mathsf{\uT}}}
\newcommand{\vbMethTL}{\vbName_{\mathsf m}}

\newcommand{\REvCtxT}{{\mathcal R}} 

\newcommand{\turnsG}[1]{\, \vdash_{#1} \,}

\newcommand{\trans}[4]{#2 \turnsG{#1} #3 \leadsto #4}
\newcommand{\transGrayBox}[4]{#2 \turnsG{#1} #3 \GrayBox{\leadsto #4}}

\newcommand{\tdExpTrans}[3]{\trans{\mathsf{exp}}{#1}{#2}{#3}}
\newcommand{\tdExpTransGrayBox}[3]{\transGrayBox{\mathsf{exp}}{#1}{#2}{#3}}
\newcommand{\tdMethTrans}[3]{\trans{\mathsf{meth}}{#1}{#2}{#3}}
\newcommand{\tdMethTransGrayBox}[3]{\transGrayBox{\mathsf{meth}}{#1}{#2}{#3}}
\newcommand{\tdProgTrans}[2]{ \turnsG{\mathsf{prog}} #1 \leadsto #2}
\newcommand{\tdProgTransGrayBox}[2]{ \turnsG{\mathsf{prog}} #1 \GrayBox{\leadsto #2}}

\newcommand{\tdUpcast}[3]{\trans{\mathsf{iCons}}{#1}{#2}{#3}}
\newcommand{\tdUpcastGrayBox}[3]{\transGrayBox{\mathsf{iCons}}{#1}{#2}{#3}}
\newcommand{\tdDowncast}[3]{\trans{\mathsf{iDestr}}{#1}{#2}{#3}}
\newcommand{\tdDowncastGrayBox}[3]{\transGrayBox{\mathsf{iDestr}}{#1}{#2}{#3}}

\newcommand{\distinct}[1]{{\mathem distinct}(#1)}

\newcommand{\methodLookup}[2]{{\mathem methodLookup}(#1, #2)}

\newcommand{\mapPerm}{\pi}               

\newcommand{\redLRk}[5]{#2 \turnsSomethingk{red-rel}{#1} #4 \approx_{#3} #5}

\newcommand\Sem[1]{\llbracket#1\rrbracket}
\newcommand{\redLRConfk}[4]{ #3 \approx #4 \in \Sem{#2}_{#1}^{\pair{\foreachN{D}}{\vbMethTL}}}

\newcommand{\RepK}[1]{{K_{\mathit{Rep}}}_{#1}}

\newcommand{\mT}[2]{\xT_{{#1},{#2}}}

\newcommand\sSynCatName[1]{\mbox{\small #1}}

\newcommand\sTransSection[2]{%
  \begin{tabularx}{\textwidth}{Xr}%
    #1\hfill & {\small \it #2}%
  \end{tabularx}%
}

\newcommand\sOver[2][!*NEVER USED ARGUMENT*!]{%
  \ifthenelse{\equal{#1}{!*NEVER USED ARGUMENT*!}}{\overline{#2}}{\overline{#2}^{#1}}%
}

\newcommand\sOptEmpty\bullet

\newcommand\RefTirName[1]{\TirName{#1}}
\newcommand\Rule[1]{\RefTirName{#1}}

\newcommand\inferruleLeft[3][]{\inferrule*[vcenter, Left=#1]{#2}{#3}}

\newcommand\GrayBox[1]{\mcolorbox{GrayBgColor}{#1}}

\newcommand{\bi}{\begin{array}[t]{@{}l@{}}}
\newcommand{\ei}{\end{array}}
\newcommand{\ba}{\begin{array}}
\newcommand{\ea}{\end{array}}
\newcommand{\bda}{\[\ba}
\newcommand{\eda}{\ea\]}
\newcommand{\bp}{\begin{quote}\tt\begin{tabbing}}
\newcommand{\ep}{\end{tabbing}\end{quote}}

\def\ruleform#1{{\setlength{\fboxrule}{0.5pt}\fbox{\normalsize \ensuremath{#1}}}}

\newcommand{\myirule}[2]{{\renewcommand{\arraystretch}{1.2}\ba{c} #1
                      \\ \hline #2 \ea}}

\newcommand{\figurebox}[1]
        {\fbox{\begin{minipage}{\textwidth} #1 \end{minipage}}}
\newcommand{\boxfig}[3]
           {\begin{figure*}\figurebox{#3}\vspace{-2ex}\caption{\label{#1}#2}\end{figure*}}

\newcommand\noteForall[1]{(\textrm{for all}~#1)}

\begin{document}

\title{A
      Dictionary-Passing Translation of Featherweight Go}

\author{Martin Sulzmann\inst1  \and Stefan Wehr\inst2}
\institute{
  Karlsruhe University of Applied Sciences, Germany\\
  \email{martin.sulzmann@h-ka.de}
  \and
  Offenburg University of Applied Sciences, Germany\\
  \email{stefan.wehr@hs-offenburg.de}
}

\maketitle

\begin{abstract}
  The Go programming language is an increasingly popular language but some of
  its features lack a formal investigation.
  This article explains Go's resolution mechanism for overloaded methods and
  its support for structural subtyping by
  means of translation from Featherweight Go to a simple target language.
  The translation employs a form of dictionary passing known from type classes
  in Haskell and preserves the
  dynamic behavior of Featherweight Go programs.
\end{abstract}

\section{Introduction}

The Go programming language~\cite{golang}, introduced by Google in 2009,
is syntactically close to C and incorporates features
that are well-established in other programming languages.
For example, a garbage collector as found in Java~~\cite{java16}, built-in support for concurrency
and channels in the style of Concurrent ML~\cite{cml-design-etc},
higher-order and anonymous functions known from functional languages such as Haskell~\cite{haskell2010}.
Go also supports method overloading for structures where related methods can be grouped together
using interfaces.
Unlike Java, where subtyping is nominal, Go supports structural subtyping among interfaces.

Earlier work by Griesmer and co-authors~\cite{FeatherweightGo}
introduces Featherweight Go (FG), a minimal core calculus that includes the essential features of Go.
Their work specifies static typing rules and a run-time method lookup semantics for FG.
However, the actual Go implementation appears to employ a different dynamic semantics.
Quoting Griesmer and co-workers:
\begin{quote}\it
Go is designed to enable efficient implementation. Structures are laid out in memory as a sequence of fields,
while an interface is a pair of a pointer to an underlying structure and a pointer to a dictionary of methods.
\end{quote}
To our knowledge, nobody has so far formalized such a dictionary-passing translation for FG and established its semantic
equivalence with the FG run-time method lookup dynamic semantics.
Hence, we make the following contributions:

\begin{itemize}
\item Section~\ref{sec:type-directed} specifies the translation of source FG programs to an
  untyped lambda calculus with pattern matching.
     We employ a dictionary-passing translation scheme \`a la type classes~\cite{Hall:1996:TCH:227699.227700}
     to statically resolve overloaded FG method calls. The translation
     is guided by the typing of the FG program.
   \item Section~\ref{sec:properties} establishes the semantic correctness of the dictionary-passing translation.
     The proof for this result is far from trivial.
     We require step-indexed logical relations~\cite{10.1007/11693024_6} as there can be
      cyclic dependencies between interfaces and method declarations.
\end{itemize}

Section~\ref{sec:featherweight-go} specifies Featherweight Go (FG)
and Section~\ref{sec:target-language} specifies our target language.
Section~\ref{sec:related-conc} covers related works and concludes.
The upcoming section gives an overview.

\section{Overview}

\boxfig{f:fg-equality}{Equality and ordering in FG and its translation}{
  \vspace{-1ex}
\begin{lstlisting}[name=goexample,numbers=left,escapechar=@,language=Golang,mathescape=true]
type Int  struct { val int }
type Pair struct { left  Eq; right Eq }
type Eq   interface { eq(that Eq) bool }
type Ord  interface { eq(that Eq) bool; lt(that Ord) bool }

func (this Int) eq(that Eq) bool {
      return this.val == (that.(Int)).val
}
func (this Pair) eq(that Eq) bool {
      return this.left.eq(that.(Pair).left) &&
             this.right.eq(that.(Pair).right)
}
func (this Int) lt(that Ord) bool {
      return this.eq(that) ||
             (this.val < (that.(Int)).val)
}
func main() {
  var i Int = Int{1}
  var j Int = Int{2}
  var p Pair = Pair{i, j}
  var _ bool = p.eq(p)
}
\end{lstlisting}
  \vspace{-1ex}
\begin{lstlisting}[name=hsexample,firstnumber=23,numbers=left,escapechar=@,language=myhaskell,mathescape=true]
-- Field access assuming constructors $\texttt{K}_{\textit{Int}}$ and $\texttt{K}_{\textit{Pair}}$.
val (K$_{\textit{Int}}$ y) = y
left (K$_{\textit{Pair}}$ (x,_)) = x
right (K$_{\textit{Pair}}$ (_,x)) = x

-- Interface-value construction assuming constructors $\texttt{K}_{\textit{Eq}}$, $\texttt{K}_{\textit{Ord}}$.
toEq$_{\textit{Int}}$ y = K$_{\textit{Eq}}$ (y, eq$_{\textit{Int}}$)
toEq$_{\textit{Pair}}$ y = K$_{\textit{Eq}}$ (y, eq$_{\textit{Pair}}$)
toEq$_{\textit{Ord}}$ (K$_{\textit{Ord}}$ (x,eq,_)) = K$_{\textit{Eq}}$ (x,eq)

-- Interface-value destruction.
fromEq$_{\textit{Int}}$ (K$_{\textit{Eq}}$ (K$_{\textit{Int}}$ y, _)) = K$_{\textit{Int}}$ y
fromEq$_{\textit{Pair}}$ (K$_{\textit{Eq}}$ (K$_{\textit{Pair}}$ (x,y) , _)) = K$_{\textit{Pair}}$ (x,y)
fromOrd$_{\textit{Int}}$ (K$_{\textit{Ord}}$ (K$_{\textit{Int}}$ y, _, _)) = K$_{\textit{Int}}$ y

-- Method definitions.
eq$_{\textit{Eq}}$ (K$_{\textit{Eq}}$ (x, eq)) = eq x
eq$_{\textit{Int}}$ this that = val this == val (fromEq$_{\textit{Int}}$ that)
eq$_{\textit{Pair}}$ this that = eq$_{\textit{Eq}}$ (left this) (left (fromEq$_{\textit{Pair}}$ that))
                 && eq$_{\textit{Eq}}$ (right this) (right (fromEq$_{\textit{Pair}}$ that))
lt$_{\textit{Int}}$ this that = eq$_{\textit{Int}}$ this (toEq$_{\textit{Ord}}$ that)
                 || (val this < val (fromOrd$_{\textit{Int}}$ that))
main =
 let i = K$_{\textit{Int}}$ 1
     j = K$_{\textit{Int}}$ 2
     p = K$_{\textit{Pair}}$ (toEq$_{\textit{Int}}$ i,toEq$_{\textit{Int}}$ j)
 in  eq$_{\textit{Pair}}$ p (toEq$_{\textit{Pair}}$ p)
\end{lstlisting}
  \vspace{-1ex}
}

We introduce Featherweight Go~\cite{FeatherweightGo} (FG) by an example
and then present the ideas of our dictionary-passing translation for FG.

\subsection{FG by Example}
\label{sec:fg-by-example}

FG is a syntactic subset of the full Go language, supporting structures, methods and interfaces.
The upper part in Figure~\ref{f:fg-equality}, lines 1-22, shows an example slightly
adopted from~\cite{FeatherweightGo}.
The original example covers equality in FG.
We extend the example and include an ordering relation (less or equal than) as well.

FG programs consist of a sequence of declarations defining structures, methods, interfaces and
a main function. Method bodies in FG only consist of a return statement. For clarity,
we sometimes identify subexpressions via variable bindings introduced with
\texttt{var}.
In such a declaration, the name of a variable precedes its type,
the notation \texttt{var}\,\go{_} (line 21)
indicates that we do not care about the
variable name given to the main expression. The example uses primitive types \texttt{int} and
\texttt{bool} and several operations on values of these types (\texttt{==}, \texttt{\&\&}, \ldots).
These are not part of FG.

Structures in FG are similar to structures known from C/\CC{}.
A syntactic difference is the FG convention that field names precede the types.
In FG, structures and methods are always declared separately, whereas \CC{} groups methods
together in a
class declaration. Methods in FG can be overloaded on the \emph{receiver}.
The receiver is the value on which the method operates on.

Interfaces in FG consist of a set of method declarations that share the same receiver.
For example, interface \texttt{Eq} introduces method \texttt{eq}
and interface \texttt{Ord} introduces methods \texttt{eq} and \texttt{lt} (line~3 and 4).
The (leading) receiver argument is left implicit and
method names in interfaces must always be distinct.
Interfaces are types and can be used in type declarations for structures and methods.
For example, structure \texttt{Pair} defines two  fields \texttt{left} and \texttt{right},
each of type \texttt{Eq}.
Declarations of structures must be non-recursive whereas an interface
may appear in the method declaration of the interface itself.
For example, see interface \texttt{Eq}.

FG uses the keyword \texttt{func} to introduce
methods and functions. Methods can be distinguished from ordinary functions
as the receiver argument always precedes the method name.
In FG, the only function is the main function, all other declarations introduced by
\texttt{func} are methods.

Consider the method implementation of \texttt{eq} for receiver \texttt{this} of type \texttt{Int}
starting at line 6. This definition takes care of equality among \texttt{Int} values
by making use of primitive equality \texttt{==} among \texttt{int}.
We would expect argument \texttt{that} to be of type \texttt{Int}.
However, to be able to use an \texttt{Int} value everywhere an \texttt{Eq} value is expected (to
be discussed shortly),
the signature of \texttt{eq} for \texttt{Int} must match exactly the signature declared by
interface \texttt{Eq}.
Hence, \texttt{that} has declared type \texttt{Eq}, and
we resort to a type assertion, written \texttt{that.(Int)}, to convert it to \texttt{Int}.
Type assertions involve a run-time check that may fail.
The same observation applies to the implementation of \texttt{Eq} for receiver \texttt{Pair}
(line 9).

FG supports structural subtyping among structures and interfaces.
A structure is a subtype of an interface if the structure implements the methods as declared by the interface.
For example, \texttt{Int} and \texttt{Pair} both implement interface \texttt{Eq}.
This implies the structural subtype relations
(1) $\texttt{Int} \subtypeOfSym\ \texttt{Eq}$ and
(2) $\texttt{Pair} \subtypeOfSym\ \texttt{Eq}$.
Relation (1) ensures that the construction of the pair at line 20 type checks:
variables \texttt{i} and \texttt{j} have type \texttt{Int} but can also be viewed
as type \texttt{Eq} thanks to structural subtyping.
Relation (2) resolves the method call \texttt{p.eq(p)} at line 21
as the \texttt{Pair} variable \texttt{p} also has the type \texttt{Eq}.
The method definition starting at line 9 will be chosen.

An interface $I$ is a structural subtype of another interface $J$ if
$I$ contains all of $J$'s method declarations.
For example, the set of methods of interface \texttt{Ord} is a superset of the method set of \texttt{Eq}.
This implies (3) $\texttt{Ord} \subtypeOfSym\ \texttt{Eq}$, which
is used in the method implementation of \texttt{lt} for receiver type \texttt{Int}.
See line 14 where (3) yields that variable \texttt{that} with declared
type \texttt{Ord} also has type \texttt{Eq}.
Thus, the method call \texttt{this.eq(that)} is resolved via
the method definition from line 6.

\subsection{Dictionary-Passing Translation}

We translate FG programs by applying a form of dictionary-passing translation known from type classes~\cite{Hall:1996:TCH:227699.227700}.
As our target language we consider an untyped functional language with pattern matching where
we use Haskell style syntax for expressions and patterns.
Each FG interface translates to a pair consisting of a structure value and a dictionary.
The dictionary holds the set of methods available as specified by the interface
whereas the structure implements these methods.
We refer to such pairs as \emph{interface-values}.
The translation is type-directed as we need type information to resolve method calls and construct the
appropriate dictionaries and interface values.

Lines 23-49 show the result of applying our dictionary-passing\footnote{Technically, we are passing around
  interface-values wrapping dictionaries of methods.} translation scheme to the FG program (lines 1-22).
We use a tagged representation to encode FG structures in the target language.
Hence, for each structure $S$, we assume a data constructor \texttt{K$_S$},
where we use pattern matching to represent field access
(lines 23-26). For example, structure \texttt{Pair} implies the data constructor \texttt{K$_{\textit{Pair}}$}.
For convenience, we assume tuples and make use of don't care patterns \verb|_|.

A method call on an interface type translates to a lookup of the method in the dictionary
of the corresponding interface-value.
Like structures, interface-values are tagged in the target language.
For example, line 39 introduces the helper function \texttt{eq$_{\textit{Eq}}$} to perform method lookup
for method \texttt{eq} of interface \texttt{Eq}. The constructor for an \texttt{Eq} interface-value is
\texttt{K$_{\textit{Eq}}$}. Hence,
we pattern match on \texttt{K$_{\textit{Eq}}$} and extract the underlying structure value and method definition.
A method call such as \texttt{this.left.eq(}$\ldots$\texttt{)} in the source program (line 10) with receiver \texttt{this.left}
of type \texttt{Eq} then translates
to \texttt{eq$_{\textit{Eq}}$\,(left this)}\,$\ldots$ (line 41).

A method call on a structure translates to the method definition for this receiver type.
For example, we write \texttt{eq$_{\textit{Int}}$} to refer to the translation
of the method definition of
\texttt{eq} for receiver type \texttt{Int}.
A method call such as \texttt{this.eq(}$\ldots$\texttt{)} in the source program (line 14) with
receiver \texttt{this} of type \texttt{Int} then translates
to \texttt{eq$_{\textit{Int}}$\,this}\,$\ldots$ (line 43).

The construction of interface-values is based on structural subtype relations.
Recall the three structural subtype relations we have seen earlier:
(1) $\texttt{Int} \subtypeOfSym\ \texttt{Eq}$ and
(2) $\texttt{Pair} \subtypeOfSym\ \texttt{Eq}$
and (3) $\texttt{Ord} \subtypeOfSym\ \texttt{Eq}$.
Relation (1) implies the interface-value constructor \texttt{toEq$_{\textit{Int}}$} (line 29),
which builds an \texttt{Eq} interface-value via the given structure value \texttt{y} and a dictionary
consisting only of the method \texttt{eq$_{\textit{Int}}$}.
Relation (2) implies a similar interface-value constructor \texttt{toEq$_{\textit{Pair}}$} (line 30).
Relation (3) gives raise to the interface-value constructor \texttt{toEq$_{\textit{Ord}}$} (line 31),
which transforms some \texttt{Ord}
into an \texttt{Eq} interface-value.
We assume that in case a dictionary consists of several methods, methods are kept in fixed order.

Type assertions imply interface-value destructors.
For example, the source expression \texttt{that.(Int)} (line 7) performs a run-time check,
asserting that \texttt{that} has type \texttt{Int}.
In terms of the dictionary-passing translation, function \texttt{fromEq$_{\textit{Int}}$} (line 34)
performs this check.
Via the pattern \texttt{K$_{\textit{Eq}}$ (K$_{\textit{Int}}$ y,)}\verb| _)|,
we assert that the underlying target structure must result from \texttt{Int}.
If the interface-value contains a value not tagged with \texttt{K$_{\textit{\textit{Int}}}$}, the pattern
matching fails at run-time, just as the type assertion in FG.
Interface-value destructors \texttt{fromEq$_{\textit{Pair}}$} and \texttt{fromOrd$_{\textit{Int}}$}
(lines~35, 36) result from similar uses of type assertions.

To summarize, each use of structural subtyping implies a interface-value constructor being inserted
in the target program. For example, typing the source expression
\texttt{p.eq(p)} in line 21
relies on structural subtyping $\texttt{Pair} \subtypeOfSym \texttt{Eq}$ because argument \texttt{p}
has type \texttt{Pair} but method \texttt{eq} requires a parameter of type \texttt{Eq}.
Thus, the translation of this expression is
\texttt{eq$_{\textit{Pair}}$ p (toEq$_\textit{Pair}$ p)} in line 49.

Similarly, type assertions imply interface-value destructors.
For example, the source expression
\texttt{that.(Pair).left} in line 10
use a type assertion on \texttt{that}, which has type \texttt{Eq}.
Thus, it translates to the target expression
\texttt{left (fromEq$_{\textit{Pair}}$ that)} in line 41.

We continue by introducing FG and our target language followed by the full details of the
dictionary-passing translation.

\section{Featherweight Go}
\label{sec:featherweight-go}

\boxfig{f:fg}{Featherweight Go (FG)}{

\bda{ll}
{
  \ba{ll}
  \GoSynCatName{Field name}     & f
\\ \GoSynCatName{Method name}   & m
\\ \GoSynCatName{Variable name} & x,y
\\ \GoSynCatName{Structure type name}    & t_S, u_S
\\ \GoSynCatName{Interface type name}    & t_I, u_I
\\ \GoSynCatName{Type name}              & t, u ::= t_S \mid t_I
\\ \GoSynCatName{Method signature}       & M ::= (\overline{x_i \ t_i}) \ t
\\ \GoSynCatName{Method specification}   & R, S ::= m M
\ea
}

&

\ba{llc}
\GoSynCatName{Expression} & d,e ::=
\\ \quad \GoSynCatName{Variable}          & \quad x & \mid
\\ \quad \GoSynCatName{Method call}       & \quad e.m(\overline{e}) & \mid
\\ \quad \GoSynCatName{Structure literal} & \quad t_S \{ \overline{e} \} & \mid
\\ \quad \GoSynCatName{Select}            & \quad e.f & \mid
\\ \quad \GoSynCatName{Type assertion}    & \quad e.(t)
\ea

\eda
\vspace{-1ex}
\bda{lcl}

\ba{llc}
\GoSynCatName{Type literal} & L ::=
\\ \quad \GoSynCatName{Structure} & \quad \STRUCT\ \{ \overline{f \ t} \} & \mid
\\ \quad \GoSynCatName{Interface} & \quad \INTERFACE\ \{ \overline{S} \}
\ea

&
\gap
&

\ba{llc}
\GoSynCatName{Declaration} & D ::=
\\ \quad \GoSynCatName{Type}   & \quad \TYPE\ t \ L & \mid
\\ \quad \GoSynCatName{Method} & \quad \FUNC\ (x \ t_S) \ mM \ \{ \RETURN\ e \}
\ea
\eda
\vspace{-1ex}
\bda{llcl}
\GoSynCatName{Program} & P & ::= & 
          \overline{D} \ \FUNC\ \MAIN () \{ \_ = e \}
\eda

\sTransSection{
  \ruleform{\fgSub{\foreachN{D}}{t}{u}}
}{Subtyping}
  \begin{mathpar}
    \inferrule[methods-struct]
  {}
  {
    \methodSpecifications{\foreachN{D}}{t_S} = \{ m M \mid \FUNC\ (x \ t_S) \ mM \ \{ \RETURN\ e \} \in  \foreachN{D} \}
  }

  \inferrule[methods-iface]{
    \TYPE\ t_I \ \INTERFACE\ \{ \foreachN{S} \} \in \foreachN{D}
  }
  { \methodSpecifications{\foreachN{D}}{t_I} = \{ \foreachN{S} \}
  }

  \inferrule[sub-struct-refl]{
   }
  { \fgSub{\foreachN{D}}{t_S}{t_S}
  }

  \inferrule[sub-iface]{
   \methodSpecifications{\foreachN{D}}{t} \supseteq \methodSpecifications{\foreachN{D}}{u_I}
  }
  { \fgSub{\foreachN{D}}{t}{u_I}
  }

  \end{mathpar}
\sTransSection{
  \ruleform{\foreachN{D} \turnsFG \reduce{d}{e}}
}{
  Reductions
}
  \bda{llrl}
  \sSynCatName{Value} &
  v & ::= & t_S \{ \foreachN{v} \}
  \\
  \sSynCatName{Evaluation context} &
  \EvCtx & ::= & []
   \mid t_S \{ \foreachN{v}, \EvCtx, \foreachN{e} \}
   \mid \EvCtx.f
   \mid \EvCtx.(t)
   \mid \EvCtx.m(\foreachN{e})
   \mid v.m(\foreachN{v}, \EvCtx, \foreachN{e})
   \\
      \sSynCatName{Substitution (FG values)} &
   \vbFG & ::= & \Angle{\foreachN{\subst{x_i}{v_i}}}
\eda
\begin{mathpar}
 \inferrule[fg-context]
 {\foreachN{D} \turnsFG \reduce{d}{e}
 }
 {\foreachN{D} \turnsFG \reduce{\EvCtx [d]}{\EvCtx [e]}
 }

 \inferrule[fg-field]
 { \TYPE\ t_S \ \STRUCT\ \{ \foreachN{f \ t} \} \in  \foreachN{D}
 }
 { \foreachN{D} \turnsFG \reduce{t_S \{ \foreachN{v} \}.f_i}{v_i}
 }

 \inferrule[fg-call]
  { v = t_S \{ \foreachN{v} \}
    \\ \FUNC\ (x \ t_S) \ m (\foreachN{x \ t}) \ t \ \{ \RETURN\ e \} \in \foreachN{D}
 }
 {\foreachN{D} \turnsFG \reduce{v.m(\foreachN{v})}{\Angle{\subst{x}{v}, \foreachN{\subst{x_i}{v_i}}} e}
 }

 \inferrule[fg-assert]
           { v = t_S \{ \foreachN{v} \}
              \\ \fgSub{\foreachN{D}}{t_S}{t}
 }
 { \foreachN{D} \turnsFG \reduce{v.(t)}{v}
 }
\end{mathpar}

}

Featherweight Go (FG)~\cite{FeatherweightGo} is a tiny fragment of Go containing
only structures, methods and interfaces.
Figure~\ref{f:fg} gives the syntax of FG.
With the exception of variable bindings in function bodies, the primitive type \texttt{int} with operations
\texttt{==} and \texttt{<}, and the primitive type \texttt{bool} with operations
\texttt{\&\&} and \texttt{||},
we can represent the example from Figure~\ref{f:fg} in FG.
Compared to the original presentation of FG~\cite{FeatherweightGo} we use symbol $L$ instead of $T$ (for type literals),
and omit the \PACKAGE\ keyword at the start of a FG program.
Overbar notation $\sOver[n]{\xi}$ denotes the sequence $\xi_1 \ldots \xi_n$ for some
syntactic construct $\xi$, where in some places commas separate the sequence
items.
If irrelevant, we omit the $n$ and simply write
$\sOver{\xi}$. Using the index variable $i$ under an overbar
marks the parts that vary from sequence item to sequence
item; for example, $\sOver[n]{\xi'\,\xi_i}$ abbreviates
$\xi'\,\xi_1\ldots\xi'\,\xi_n$
and $\sOver[q]{\xi_j}$ abbreviates
$\xi_{j1}\,\ldots\,\xi_{jq}$.

FG is a statically typed language.
For brevity, we omit a detailed description of the FG typing rules
as they will show up in the type-directed translation.
The following conditions must be satisfied.
\begin{description}
\item[\FGStructNonRec:] Structures must be non-recursive.
\item[\FGUniqueFields:] For each struct, field names must be distinct.
\item[\FGUniqueMethSpec:] For each interface, method names must be distinct.
\item[\FGUniqueReceiver:] Each method declaration is uniquely identified by the receiver type and method name.
\end{description}

FG supports structural subtyping, written $\fgSub{\foreachN{D}}{t}{u}$.
A struct $t_S$ is a subtype of an interface $t_I$ if $t_S$ implements all the methods
specified by the interface $t_I$. An interface $t_I$ is a subtype of another interface $u_I$ if
the methods specified by $t_I$ are a superset of the methods specified by $u_I$.
The structural subtyping relations are specified in the middle part of Figure~\ref{f:fg}.

Next, we consider the dynamic semantics of FG.
The bottom part of Figure~\ref{f:fg} specifies the reduction of FG programs by making use of structural operational semantics rules
of the form $\foreachN{D} \turnsFG \reduce{d}{e}$
to reduce expression $d$ to expression $e$ under the  sequence $\foreachN{D}$ of declarations.

Rule \Rule{fg-context} makes use of evaluation contexts with holes to apply a reduction inside an expression.
Rule \Rule{fg-field} deals with field access.
Condition \FGUniqueFields\ guarantees that field name lookup is unambiguous.
Rule \Rule{fg-call} reduces method calls.
Condition \FGUniqueReceiver\ guarantees that method lookup is unambiguous.
The method call is reduced to the method body $e$ where we map the receiver argument to a concrete value $v$
and method arguments $x_i$ to concrete values $v_i$.
This is achieved by applying the substitution $\Angle{\subst{x}{v}, \foreachN{\subst{x_i}{v_i}}}$ on $e$,
written $\Angle{\subst{x}{v}, \foreachN{\subst{x_i}{v_i}}} e$.

Rule \Rule{fg-assert} covers type assertions.
We need to check that the type $t_S$ of value $v$ is consistent with the type $t$ asserted in the program text.
If $t$ is an interface, then $t_S$ must implement all the methods as specified by this interface.
If $t$ is a struct type, then $t$ must be equal to $t_S$.
Both checks can be carried out by checking that $t_S$ and $t$ are in a structural subtype relation.

We write $\foreachN{D} \turnsFG \reduceN{e}{v}$ to denote that under
the declarations $\foreachN{D}$, expression $e$ reduces to the value $v$
in a finite number of steps.
We write $\reduceFGk{k}{\foreachN{D}}{e}{v}$ to denote that under
the declarations $\foreachN{D}$, expression $e$ reduces to the value $v$
within at most $k$ steps. This means we might need fewer than $k$ steps
but $k$ are clearly sufficient to reduce the expression to some value.
If there is no such $v$ for any number of steps,
we say that $e$ is \emph{irreducible} w.r.t.~$\foreachN{D}$, written $\irred(\foreachN{D},e)$.

\section{Target Language}
\label{sec:target-language}

\boxfig{f:target-lang}{Target Language (TL)}{
  \vspace{-2ex}
  \bda{lcr}
  {
    \ba{lll}
    \sSynCatName{Expression} & \expT ::=
    \\ \quad \sSynCatName{Variable}    & \quad \xT \mid \yT         & \mid
    \\ \quad \sSynCatName{Constructor} & \quad \kT                  & \mid
    \\ \quad \sSynCatName{Application} & \quad \expT \ \expT        & \mid \
    \\ \quad \sSynCatName{Abstraction} & \quad \lambda \xT. \expT   & \mid
    \\ \quad \sSynCatName{Pattern case} & \quad \tcaseof{\expT}{[\foreachN{\clsT}]}
    \ea
  }
  &~\quad~&
  {
    \ba{lrl}
     \sSynCatName{Pattern clause} &
     \clsT & ::=  \patT \rightarrow \expT
     \\
     \sSynCatName{Pattern} &
     \patT & ::= \kT \ \foreachN{\xT}
     \\
    \sSynCatName{Program}  &
    \program & ::= {\begin{array}[t]{l}\arraycolsep=0pt
      \LET\ \foreachN{\yT_i = \lambda \xT_i . \expT_i}\\
      \IN\ \expT
      \end{array}}
     \ea
   }
   \eda
  \bda{llrl}
  \sSynCatName{TL values} &
  \uT & ::= & \xT  \mid \kT\ \foreachN{\uT}
  \\
  \sSynCatName{TL evaluation context} &
  \REvCtxT & ::= & []
   \mid \kT\  \foreachN{\uT} \REvCtxT \foreachN{\expT}
   \mid \tcaseof{\REvCtxT}{[\foreachN{\patT \rightarrow \expT}]}
   \mid \REvCtxT\ \expT
   \mid \uT\ \REvCtxT
  \\
     \sSynCatName{Substitution (TL values)} &
   \vbTL & ::= & \Angle{\foreachN{\subst{\xT}{\uT}}}
  \\
     \sSynCatName{Substitution (TL methods)} &
   \vbMethTL & ::= & \Angle{\foreachN{\subst{\yT}{\lambda \xT. \expT}}}
   \eda

\vspace{1ex}
\sTransSection{
  \ruleform{\reduceTL{\vbMethTL}{\expT}{\expT'}}
}{
  TL expression reductions
}
\vspace{-0.5ex}
\begin{mathpar}
 \inferrule[tl-context]
 {\reduceTL{\vbMethTL}{\expT}{\expT'}
 }
 {\reduceTL{\vbMethTL}{\REvCtxT [\expT]}{\REvCtxT [\expT']}
 }

 \inferrule[tl-lambda]
 {}
 { \reduceTL{\vbMethTL}{(\lambda \xT. \expT) \ \uT}{\Angle{\subst{\xT}{\uT}} \expT}
 }

 \inferrule[tl-case]
 { \kT \ \foreach{\xT_i}{n} \rightarrow \expT' \in [\foreachN{\patT \rightarrow \expT}]
 }
 {\reduceTL{\vbMethTL}{\tcaseof{\kT \ \foreach{\uT_i}{n}}{{[\foreachN{\patT \rightarrow \expT}]}}}
         {\Angle{\foreach{\subst{\xT_i}{\uT_i}}{n}} \expT'}
 }

 \inferrule[tl-method]
 {}
 { \reduceTL{\vbMethTL}{\yT \ \expT}{\vbMethTL(\yT) \ \expT}
 }

\end{mathpar}

\vspace{1ex}
\sTransSection{
  \ruleform{\reduceTL{}{\program}{\program'}}
}{
  TL reductions
}
\vspace{-0.5ex}
\begin{mathpar}
    \inferrule[tl-prog]
 {\reduceTL{\Angle{\foreachN{\subst{\yT_i}{\lambda \xT_i. \expT_i}}}}{\expT}{\expT'}  }
 { \reduceTL{}{\tletrecin{\foreachN{\yT_i = \lambda \xT_i.\expT_i}}{\expT}}{\tletrecin{\foreachN{\yT_i = \lambda \xT_i. \expT_i}}{\expT'}}
 }

\end{mathpar}
} 

Figure~\ref{f:target-lang} specifies the syntax and dynamic semantics of our target language (\TL).
We use capital letters for constructs of the target language. Target expressions $E$ include
variables $X,Y$, data constructors $K$,
function application, lambda abstraction and case expressions to pattern match against constructors.
In a case expression with only one pattern clause,
we often omit the brackets and just write
$\tcaseof{\expT}{\patT \rightarrow \expT}$.
A program consists of a sequence of function definitions and a (main) expression.
The function definitions are the result of translating FG method definitions.

We assume data constructors for tuples up to some fixed but arbitrary size. The syntax
$(\foreach{E}{n})$ constructs an $n$-tuple when used as an expression, and deconstructs
it when used in a pattern context.
At some places, we use nested patterns as an abbreviation for nested case expressions.
The notation $\lambda \patT. \expT$ stands for
$\lambda X . \tcaseof{X}{[\patT \rightarrow \expT]}$, where $X$ is fresh.

Representing the example from Figure~\ref{f:fg} in the target language
requires some more straightforward extensions:
integers with operations
\texttt{==} and \texttt{<}, booleans with operations
\texttt{\&\&} and \texttt{||}, let-bindings inside expressions,
and top-level bindings. The target language can encode the last two features
via lambda-abstractions and
top-level let-bindings.

The structural operational semantics employs two types of substitutions.
Substitution $\vbTL$ records the bindings resulting from pattern matching and function applications.
Substitution $\vbMethTL$ records the bindings for translated method definitions (i.e. for
top-level let-bindings).
Target values consist of constructors and variables.
A variable may be a value if it refers to a yet to be evaluated method binding.

Reduction of programs is mapped to reduction of expressions under a method substitution.
See rule \Rule{tl-prog}.
The remaining reduction rules are standard.

We write $\reduceTLN{\vbMethTL}{\expT}{\uT}$  to denote that under
substitution $\vbMethTL$, expression $\expT$ reduces to the value $\uT$
in a finite number of steps.
We write $\reduceTLk{k}{\vbMethTL}{\expT}{\uT}$ to denote that under substitution $\vbMethTL$,
expression $\expT$ reduces to $\uT$
within at most $k$ steps. This means we might need fewer than $k$ steps
but $k$ are clearly sufficient.
If there is no such $\uT$ for any number of steps,
we say that $\expT$ is \emph{irreducible} w.r.t.~$\vbMethTL$, written $\irred(\vbMethTL,\expT)$.

\section{Dictionary-Passing Translation}
\label{sec:type-directed}

We formalize the dictionary-passing translation of FG to TL.
The translation rules are split over two figures.
Figure~\ref{f:type-directed-meth-prog-dec} covers methods, programs and some expressions.
Figure~\ref{f:upcast-downcast} covers structural subtyping and type assertions.
The translation rules are guided by type checking the FG program.
The gray shaded parts highlight target terms that are generated.
If these parts are ignored, the translation rules are effectively equivalent
to the FG type checking rules~\cite{FeatherweightGo}.
We assume that conditions FG1-4 hold as well.

We use the following conventions.
We assume that each FG variable $x$ translates to the TL variable $\xT$.
For each structure $t_S$ we introduce a TL constructor $K_{t_S}$.
For each interface $t_I$ we introduce a TL constructor $K_{t_I}$.
In the translation, a source value of an interface type $t_I$ translates to an interface-value tagged by $K_{t_I}$.
The interface-value contains the underlying structure value and a dictionary consisting of the set of methods as specified by the interface.
For each method declaration $\FUNC\ (x \ {t_S}) \ m M \ \{ \RETURN\ e \}$ we introduce
a TL variable $\mT{m}{t_S}$, thereby relying on \FGUniqueReceiver{} which guarantees that $m$ and
$t_S$ uniquely identify this declaration.
We write $\fgEnv$ to denote typing environments where we record the types of FG variables.
The notation $[n]$ is a short-hand for the set $\{1,\dots,n\}$.

\subsection{Translating programs, methods and expressions}

\boxfig{f:type-directed-meth-prog-dec}{Translation of methods, programs and expressions}{
  Convention for mapping source to target terms
  \vspace{-1ex}
  \begin{mathpar}
    x \leadsto \xT
    \quad
    t_S \leadsto K_{t_S}
    \quad
    t_I \leadsto K_{t_I}
    \quad
    \FUNC\ (x \ {t_S}) \ m M \ \{ \RETURN\ e \} \leadsto \mT{m}{t_S}
  \end{mathpar}
  \bda{llcl}
\mbox{FG Environment} & \fgEnv & ::= & \{ \} \mid \{ x : t \} \mid \fgEnv \cup \fgEnv
  \eda
  \vspace{0.5ex}

\sTransSection{
  \ruleform{\tdMethTransGrayBox{\foreachN{D}}{\FUNC\ (x \ t_S) \ m (\foreachN{x \ t}) \ t}{\expT} }
}{
  Translating method declarations
}
  \bda{c}
  \inferruleLeft[td-method]
         { \distinct{x,\foreach{x}{n}}
           \\ \tdExpTransGrayBox{\pair{\foreachN{D}}{\{ x : t_S, \foreach{x_i : t_i}{n} \}}}{e : t}{\expT}
          }
          { \tdMethTransGrayBox{\foreachN{D}}
                        {\FUNC\ (x \ t_S) \ m (\foreach{x \ t}{n}) \ t \ \{ \RETURN\ e \}}
                        {\lambda \xT . \lambda (\foreach{\xT }{n}). \expT}
          }

  \eda

\sTransSection{
  \ruleform{\tdProgTransGrayBox{P}{\program}}
}{
  Translating programs
}
  \bda{c}
 \inferrule[td-prog]
 {\tdExpTransGrayBox{\pair{\foreachN{D}}{\EmptyFgEnv}}{e : t}{\expT} \\\\
   \tdMethTransGrayBox{\foreachN{D}}{D_i'}{\expT_i}\\
   \\D_i' = \FUNC\ (x_i \ {t_S}_i) \ m_i M_i \ \{ \RETURN\ e_i \}\\
   (\textrm{for all}~i \in [n],
   \textrm{where}~\foreach{D'}{n}~\textrm{are the}~\FUNC~\textrm{declarations in}~\foreachN{D})\\
 }
 { \tdProgTransGrayBox{\ignore{\PACKAGE\ \MAIN;} \foreachN{D} \ \FUNC\ \MAIN () \{ \_ = e \}}
   { \LET\ \foreach{\mT{m_i}{{t_S}_i} = \expT_i}{n} \ \IN\ \expT}
 }
  \eda

  \sTransSection{
    \ruleform{\tdExpTransGrayBox{\pair{\foreachN{D}}{\fgEnv}}{e : t}{\expT}}
  }{Translating expressions}
  \begin{mathpar}
  \inferrule[td-var]
            { (x : t) \in \fgEnv
            }
            { \tdExpTransGrayBox{\pair{\foreachN{D}}{\fgEnv}}{x : t}{\xT}
            }

  \inferrule[td-struct]
  {\TYPE\ t_S \ \STRUCT\ \{ \foreach{f \ t}{n} \} \in  \foreachN{D}
               \\   \tdExpTransGrayBox{\pair{\foreachN{D}}{\fgEnv}}{e_i : t_i}{\expT_i}\quad\noteForall{i \in [n]}
            }
            { \tdExpTransGrayBox{\pair{\foreachN{D}}{\fgEnv}}{t_S \{ \foreach{e}{n} \} : t_S }{\kT_{t_S} \ (\foreach{\expT}{n})}
            }

  \inferrule[td-access]
  {\tdExpTransGrayBox{\pair{\foreachN{D}}{\fgEnv}}{e : t_S}{\expT}
             \\  \TYPE\ t_S \ \STRUCT\ \{ \foreach{f \ t}{n} \} \in  \foreachN{D}
          }
           { \tdExpTransGrayBox{\pair{\foreachN{D}}{\fgEnv}}
                        {e.f_i : t_i }
                        { \CASE\ \expT\ \OF\ \kT_{t_S} \ (\foreach{\xT}{n}) \rightarrow \xT_i}
           }

  \inferrule[td-call-struct]
            { m (\foreach{x \ t}{n}) \ t \in \methodSpecifications{\foreachN{D}}{t_S}
              \\
              \tdExpTransGrayBox{\pair{\foreachN{D}}{\fgEnv}}{e : t_S}{\expT}
              \\ \tdExpTransGrayBox{\pair{\foreachN{D}}{\fgEnv}}{e_i : t_i}{\expT_i}\quad\noteForall{i \in [n]}
          }
          { \tdExpTransGrayBox{\pair{\foreachN{D}}{\fgEnv}}{e.m(\foreach{e}{n}) : t}{\mT{m}{t_S} \ \expT \ (\foreach{\expT}{n}) } }

  \inferrule[td-call-iface]
  {\tdExpTransGrayBox{\pair{\foreachN{D}}{\fgEnv}}{e : t_I}{\expT}
    \\ \TYPE\ t_I \ \INTERFACE\ \{ \foreachN{S} \} \in \foreachN{D}
    \\ S_j = m(\foreach{x \ t}{n})\,t
    \\ \tdExpTransGrayBox{\pair{\foreachN{D}}{\fgEnv}}{e_i : t_i}{\expT_i}\quad\noteForall{i \in [n]}
    \\ X,\foreach{X}{q}\textrm{~fresh}
  }
  { \tdExpTransGrayBox{\pair{\foreachN{D}}{\fgEnv}}
    {e.m(\foreach{e}{n}) : t}
    {\CASE\ \expT \ \OF\ \kT_{t_I} \ (\xT, \foreach{\xT}{q}) \rightarrow \xT_j \ \xT \ (\foreach{\expT}{n}) }
  }
\end{mathpar}

}

\boxfig{f:upcast-downcast}{Translation of structural subtyping and type assertions}{

    \sTransSection{
    \ruleform{\tdExpTransGrayBox{\pair{\foreachN{D}}{\fgEnv}}{e : t}{\expT}}
  }{Translating structural subtyping and type assertions}
  \begin{mathpar}
  \inferrule[td-sub]
            {\tdExpTransGrayBox{\pair{\foreachN{D}}{\fgEnv}}{e : t}{\expT_2}
          \\ \tdUpcastGrayBox{\foreachN{D}}{\subtypeOf{t}{u}}{\expT_1}
          }
            { \tdExpTransGrayBox{\pair{\foreachN{D}}{\fgEnv}}{e : u}{\expT_1 \ \expT_2} }

  \inferrule[td-assert]
            {\tdExpTransGrayBox{\pair{\foreachN{D}}{\fgEnv}}{e : t_I}{\expT_2}
          \\ \tdDowncastGrayBox{\foreachN{D}}{\assertOf{t_I}{u}}{\expT_1}
          }
          { \tdExpTransGrayBox{\pair{\foreachN{D}}{\fgEnv}}{e.(u) : u}{\expT_1 \ \expT_2} }

  \end{mathpar}

  \sTransSection{
  \ruleform{\tdUpcastGrayBox{\foreachN{D}}{\subtypeOf{t}{u_I}}{\expT}}
}{
  Interface-value construction
}
\begin{mathpar}
\inferrule[td-cons-struct-iface]
          {\TYPE\ t_I \ \INTERFACE\ \{ \foreachN{S} \} \in \foreachN{D}
        \\ \methodSpecifications{\foreachN{D}}{t_S} \supseteq \foreachN{S}
        \\ \foreachN{S} = \foreach{m M}{n}
          }
          {\tdUpcastGrayBox{\foreachN{D}}{\subtypeOf{t_S}{t_I}}
                    {\lambda \xT.
                      \kT_{t_I} \ (\xT, \foreach{\mT{m_i}{t_S}}{n}) }
         }

\inferrule[td-cons-iface-iface]
          {\TYPE\ t_I \ \INTERFACE\ \{ \foreach{R}{n} \} \in \foreachN{D}
       \\ \TYPE\ u_I \ \INTERFACE\ \{ \foreach{S}{q} \} \in \foreachN{D}
       \\  S_i = R_{\mapPerm(i)} \quad\noteForall{i \in [q]}
        }
          {\tdUpcastGrayBox{\foreachN{D}}{\subtypeOf{t_I}{u_I}}
                    { \lambda \xT. \CASE\,\xT\,\OF\
                          \kT_{t_I} \ (\xT, \foreach{\xT}{n})
                          \rightarrow
                          \kT_{u_I} \ (\xT, \xT_{\mapPerm(1)}, \ldots, \xT_{\mapPerm(q)})
                    }
         }
\end{mathpar}

\sTransSection{
  \ruleform{\tdDowncastGrayBox{\foreachN{D}}{\assertOf{t_I}{u}}{E}}
}{
  Interface-value destruction
}
\begin{mathpar}
  \inferrule[td-destr-iface-struct]
            {\TYPE\ t_I \ \INTERFACE\ \{ \foreach{R}{n} \} \in \foreachN{D}
             \\ \fgSub{\foreachN{D}}{t_S}{t_I}
            }
            {\tdDowncastGrayBox{\foreachN{D}}{\assertOf{t_I}{t_S}}{\lambda \xT.  \CASE\,\xT\,\OF\
                    \kT_{t_I} \ (K_{t_S} \ Y, \foreach{\xT}{n})
                      \rightarrow K_{t_S} \, Y
                   }
            }

\inferrule[td-destr-iface-iface]
{ X, Y, Y', \foreach{X}{n}~\textrm{fresh}\\
  \TYPE\ t_I \ \INTERFACE\ \{ \foreach{R}{n} \} \in \foreachN{D}
            \\
            \GrayBox{
              \textrm{for all}~\TYPE\ t_{Sj}\ \STRUCT\ \{ \foreachN{f \ u} \} \in \foreachN{D}
              ~\textrm{with}~ \tdUpcast{\foreachN{D}}{\subtypeOf{t_{Sj}}{u_I}}{\expT_j}
              \textrm{:}
            }
            \\
            \GrayBox{
              \clsT_j =
              \kT_{t_{Sj}} \ Y' \rightarrow (\expT_j \ (\kT_{t_{Sj}} \ Y'))
            }
          }
          { \tdDowncastGrayBox{\foreachN{D}}{\assertOf{t_I}{u_I}}{\lambda \xT.
                             \CASE\ \xT \ \OF\
                             \kT_{t_I}\, (\yT, \foreach{\xT}{n}) \rightarrow \CASE\ \yT  \ \OF\
                             [\foreachN{\clsT}]}
          }
\end{mathpar}
}

The translation of programs and methods boils down to the translation of expressions involved.
Rule \Rule{td-method} translates a specific method declaration,
rule \Rule{td-prog} collects all method declarations and also translates the main expression.
See Figure~\ref{f:type-directed-meth-prog-dec}.

The translation rules for expressions are of the form
$\tdExpTrans{\pair{\foreachN{D}}{\fgEnv}}{e : t}{\expT}$
where $\foreachN{D}$ refers to the sequence of FG declarations,
$\fgEnv$ refers to type binding of local variables,
$e$ is the to be translated FG expression,
$t$ its type and $\expT$ the resulting target term.
Departing from FG's original typing rules~\cite{FeatherweightGo},
the translation rules are non-syntax directed due the structural subtyping rule \Rule{td-sub}
defined in Figure~\ref{f:upcast-downcast}.
We could integrate this rule via the other rules but this would make all
the rules harder to read. Hence, we prefer to have a separate rule \Rule{td-sub}.

We now discuss the translations rules for the expression forms in
Figure~\ref{f:type-directed-meth-prog-dec}. (The remaining expression forms
are covered in Figure~\ref{f:upcast-downcast}, to be explained in the next section.)
Rule \Rule{td-var} translates variables and follows our convention that $x$ translates to $\xT$.
Rule \Rule{td-struct} translates a structure creation. The translated field elements $\expT_i$
are collected in a tuple and tagged via the constructor $K_{t_S}$.
Rule \Rule{td-access} uses pattern matching to capture field access in the translation.

Method calls are dealt with by rules \Rule{td-call-struct} and \Rule{td-call-iface}.
Rule \Rule{td-call-struct} covers the case that the receiver $e$ is of the structure type $t_S$.
The first precondition guarantees that an implementation for this specific method call exists.
(See Figure~\ref{f:fg} for the auxiliary $\methodSpecificationsRel$.)
Hence, we can assume that we have available a corresponding definition for $\mT{m}{t_S}$ in our translation.
The method call then translates to applying $\mT{m}{t_S}$ first on the translated receiver $\expT$,
followed by the translated arguments collected in a tuple $(\foreach{\expT}{n})$.

Rule \Rule{td-call-iface} assumes that receiver $e$ is of interface type $t_I$,
so $e$ translates to interface-value $E$.
Hence, we pattern match on $E$
to access the underlying value and the desired method in the dictionary.
We assume that the order of methods in the dictionary corresponds to the order
of method declarations in the interface.
The preconditions guarantee that $t_I$ provides a method $m$ as demanded by the method call,
where $j$ denotes the index of $m$ in interface $t_I$.

\subsection{Translating structural subtyping and type assertions}

Rule \Rule{td-sub} deals with structural subtyping and yields an interface-value constructor
derived via rules \Rule{td-cons-struct-iface} and \Rule{td-cons-iface-iface} in Figure~\ref{f:upcast-downcast}.
These rules correspond to the structural subtyping rules in Figure~\ref{f:fg}
but additionally yield an interface-value constructor.

The preconditions in rule \Rule{td-cons-struct-iface} check that structure $t_S$ implements the interface $t_I$.
This guarantees the existence of method definitions $\mT{m_i}{t_S}$.
Hence, we can construct the desired interface-value.

The preconditions in rule \Rule{td-cons-iface-iface} check that $t_I$'s methods are a superset of $u_I$'s methods.
This is done via the total function $\mapPerm : \{1,\ldots,q\} \to \{1,\ldots,n\}$ that matches each (wanted) method in $u_I$ against a (given) method in $t_I$.
We use pattern matching over the $t_I$'s interface-value to extract the wanted methods.
Recall that dictionaries maintain the order of method as specified by the interface.

Type assertions $e.(u)$ are dealt with in rule \Rule{td-assert} and translate to an interface-value destructor.
In the static semantics of FG there are two cases to consider.
Both cases assume that the expression $e$ is of some interface type $t_I$.
The first case asserts the type of a structure and the second case asserts the type of an interface.
Asserting that a structure is
of the type of another structure is not allowed in FG, because such a type assertion would never succeed.

Rule \Rule{td-destr-iface-struct} deals with the case that we assert the type of a structure $t_S$.
If $t_S$ does not implement the interface $t_I$, the assertion can never be successful.
Hence, we find the precondition $\fgSub{\foreachN{D}}{t_S}{t_I}$.
We pattern match over the interface-value that represents $t_I$ to check the underlying value matches $t_S$
and extract the value.
It is possible that some other value has been used to implement the interface-value that represents $t_I$.
In such a case, the pattern match fails and we experience run-time failure.

Rule \Rule{td-destr-iface-iface} deals with the case that we assert the type of an interface $u_I$
on a value of type $t_I$.
The outer case expression extracts the value $Y$ underlying interface-value $t_I$
(this case never fails).
We then check if we can construct an interface-value for $u_I$ via $Y$.
This is done via an inner case expression.
For each structure $t_{Sj}$ implementing $u_I$,
we have a pattern clause $\clsT_j$ that
matches against the constructor $\kT_{t_{Sj}}$ of the structure and then constructs an interface-value for $u_I$.
There are two reasons for run-time failure here.
First, $\yT$ (used to implement $t_I$) might not implement $u_I$;
that is, none of the pattern clauses $\clsT_j$ match.
Second, $[\foreachN{\clsT}]$ might be empty because
no receiver at all implements $u_I$. This case is rather unlikely
and could be caught statically.

\section{Properties}
\label{sec:properties}

\boxfig{f:reduce-rel-fg-tl}{Relating FG to \TL\ Reduction}{

  \sTransSection{
   \ruleform{\redLRConfk{k}{t}{e}{\expT}}
  }{FG expressions versus \TL\ expressions}
\begin{mathpar}
 \inferrule[red-rel-exp]
           {  \forall k_1 < k, k_2 < k, v, \uT.
             (k  - k_1 - k_2 > 0 \wedge
             \reduceFGk{k_1}{\foreachN{D}}{e}{v}
             \wedge
             \reduceTLk{k_2}{\vbMethTL}{\expT}{\uT})
             \\
             \implies \redLRConfk{k- k_1 - k_2}{t}{v}{\uT}
 }
 { \redLRConfk{k}{t}{e}{\expT}
 }
\end{mathpar}

  \sTransSection{
   \ruleform{\redLRConfk{k}{t}{v}{\uT}}
  }{FG values versus \TL\ values}
\begin{mathpar}
 \inferrule[red-rel-struct]
     { \TYPE\ t_S \ \STRUCT\ \{ \foreach{f \ t}{n} \} \in  \foreachN{D}
       \\ \forall i \in [n] . \redLRConfk{k}{t_i}{v_i}{\uT_i}
     }
     { \redLRConfk{k}
       {t_S}
       {t_S \{ \foreach{v}{n} \}}
       {\kT_{t_S} \ (\foreach{\uT}{n})}
     }

 \inferrule[red-rel-iface]
    {\uT = \kT_{u_S} \ \foreachN{\uT'}
      \\ \forall k_1 < k . \redLRConfk{k_1}{u_S}{v}{\uT}
      \\ \methodSpecifications{\foreachN{D}}{t_I} = \{ \foreach{m M}{n} \}
      \\ \forall k_2 < k, i \in [n] . \redLRConfk{k_2}
                       {m_i M_i}
                       {\methodLookup{\foreachN{D}}{(m_i,u_S)}}
                       {\uT_i}
     }
     { \redLRConfk{k}{t_I}{v}
             {\kT_{t_I} \ (\uT, \foreach{\uT}{n})}
     }

 \end{mathpar}

  \sTransSection{
    \ruleform{\redLRConfk{k}
                    {m M}
                    {\FUNC\ (x \ t_S) \ mM \ \{ \RETURN\ e \}}{\uT}}
  }{FG methods versus \TL\ methods}
 \begin{mathpar}
 \inferrule[red-rel-method]
           { \forall k' \leq k,
                v', \uT', \foreach{v_i}{n}, \foreach{\uT_i}{n}.
                    (\redLRConfk{k'}{t_S}{v'}{\uT'}
                    \wedge (\forall i \in [n]. \redLRConfk{k'}{t_i}{v_i}{\uT_i}))
         \\ \implies \redLRConfk{k'}
               {t}
               {\Angle{\subst{x}{v'},\foreach{\subst{x_i}{v_i}}{n}} e}
               {(\uT \ \uT') \ (\foreach{\uT}{n})}
 }
     { \redLRConfk{k}
       {m (\foreach{x \ t}{n}) \ t}
       {\FUNC\ (x \ t_S) \ m (\foreach{x \ t}{n}) \ t \ \{ \RETURN\ e \}}
       {\uT}
 }
\end{mathpar}

  \sTransSection{
    \ruleform{\redLRk{k}{\triple{\foreachN{D}}{\vbMethTL}{\fgEnv}}
                    {}
                    {\vbFG}
                    {\vbTL}}
  }{FG versus \TL\ value bindings}
  \begin{mathpar}
    \inferrule[red-rel-vb]
              { \forall (x : t) \in \fgEnv .
                 \redLRConfk{k}{t}{\vbFG(x)}{\vbTL(X)}
              }
              {
                \redLRk{k}{\triple{\foreachN{D}}{\vbMethTL}{\fgEnv}}
                    {}
                    {\vbFG}
                    {\vbTL}
             }
  \end{mathpar}

  \sTransSection{
    \ruleform{\redLRk{k}
                    {}
                    {}
                    {\foreachN{D}}
                    {\vbMethTL}}
  }{FG declarations versus \TL\ method bindings}
 \begin{mathpar}
   \inferrule[red-rel-decls]
             { \forall\,
               \FUNC\ (x \ t_S) \ m M \ \{ \RETURN\ e \} \in \foreachN{D} :\\
              \redLRConfk{k}
                    {m M}
                    {\FUNC\ (x \ t_S) \ mM \ \{ \RETURN\ e \}}{\mT{m}{t_S}}
             }
             { \redLRk{k}
                    {}
                    {}
                    {\foreachN{D}}
                    {\vbMethTL}
             }

 \end{mathpar}

}

\newcommand{\ttt}[1]{\mbox{\tt #1}}
\newcommand{\LR}[4]{ #3 \approx #4 \in \Sem{#2}_{#1} }
\newcommand{\LRt}[4]{ \ttt{#3} \approx \ttt{#4} \in \Sem{#2}_{#1} }

We wish to show that the dictionary-passing translation preserves the dynamic behavior of FG programs.
To establish this property we make use of (binary) logical relations~\cite{Plotkin1973,DBLP:journals/jsyml/Tait67}.
Logical relations express that related terms behave the same.
We say that source and target terms are \emph{equivalent} if they are related under the logical relation.
The goal is to show that FG expressions and target expressions resulting from the dictionary-passing translation
are equivalent.

For example, in FG the run-time value associated with an interface type is a structure
that implements the interface
whereas in our translation each interface translates to an interface-value.
To establish that a structure $t_S\{\foreachN{v}\}$ and an interface-value $K_{t_I} (\uT, \foreachN{\uT})$ are equivalent
w.r.t.~some interface $t_I$
we need to require that
\begin{itemize}
\item (Struct-I-Val-1) $t_S\{\foreachN{v}\}$ and $\uT$ are equivalent w.r.t.~$t_S$, and
\item (Struct-I-Val-2) method definitions for receiver type $t_S$ are equivalent to $\foreachN{\uT}$.
\end{itemize}

Because signatures in method specifications of an interface may refer to the interface itself,
there may be cyclic dependencies that then result in well-foundness issues of the definition of logical relations.
To solve this issue we include a step index~\cite{10.1007/11693024_6}.
We explain this technical point via the example in Figure~\ref{f:fg-equality}.
We will write $\LR{k}{t}{e}{\expT}$ to denote that FG expression $e$ and TL expression $\expT$
are in a logical relation w.r.t.~the FG type $t$, where $k$ is the step index.
Similarly, $\LR{k}{R}{\FUNC\ (x \ t_S) \ R \ \{ \RETURN\ e \}}{\uT}$ expresses
that a FG method declaration and a TL value $\vT$ are in a logical relation w.r.t.~the FG method
specification $R$.

Consider the FG expression \ttt{Int\{1\}} from example in Figure~\ref{f:fg-equality}.
When viewed at type \ttt{Eq},
our translation yields the interface-value \ttt{K$_{\textit{Eq}}$ (K$_{\textit{Int}}$ 1, eq$_{\textit{Int}}$)}.
We need to establish $\LRt{k_1}{\mathtt{Eq}}
    {Int\{1\}}
    {K$_{\textit{Eq}}$ (K$_{\textit{Int}}$ 1, eq$_{\textit{Int}}$)}$.
\bda{ll}\small
(1) & \LRt{k_1}{\mathtt{Eq}}
    {Int\{1\}}
    {K$_{\textit{Eq}}$ (K$_{\textit{Int}}$ 1, eq$_{\textit{Int}}$)}
\\[\smallskipamount]
& {
  \ba{rll}
  \mbox{if} & (2) &
  \LRt{k_2}{\mathtt{Int}}
         {Int\{1\}}
         {K$_{\textit{Int}}$ 1} \ \ \mbox{and}  
  \\[\smallskipamount]
   & (3) &
 \LRt{k_3}{\texttt{eq(y Eq)\,\gokw{bool}}}
          {\gokw{func} (x Int) eq(y Eq)\,\gokw{bool}\,\{return\,$e$\}}
          {eq$_{\textit{Int}}$}
          \\ & & \mbox{where $k_2 < k_1, k_3 < k_1$}
          \\[\smallskipamount]
  & & {
    \ba{ll}
    \mbox{if}~(4) &
    \forall \LRt{k_4}{\texttt{Int}}{$v_1$}{$\uT_1$}, \LRt{k_4}{\texttt{Eq}}{$v_2$}{$\uT_2$}.
    \\&
            \LRt{k_4}{\texttt{\gokw{bool}}}
                {$\Angle{\subst{x}{v_1},\subst{y}{v_2}}e$}
                {eq$_{\textit{Int}}$ $\uT_1$ $\uT_2$}
    \mbox{where $k_4 \leq k_3$}
    \ea
  }
  \ea
}
\eda

Following (Struct-I-Val-1) and (Struct-I-Val-2), (1) holds if we can establish (2) and (3).
(2) is easy to establish.
(3) holds if we can establish (4).
(4) states that for equivalent inputs the respective method definitions are equivalent as well.
Without the step index, establishing $\LRt{}{\texttt{Eq}}{.}{.}$
would reduce to establishing $\LRt{}{\texttt{Eq}}{.}{.}$.
We are in a cycle. With the step index,
$\LR{k_1}{\texttt{Eq}}{.}{.}$
reduces to $\LRt{k_4}{\texttt{Eq}}{.}{.}$ where $k_4 < k_1$.
The step index represents the number of reduction steps we can take
and will be reduced for each reduction step.
Thus, we can give a well-founded definition of our logical relations.

Figure~\ref{f:reduce-rel-fg-tl} gives the step-indexed logical relations
to relate FG and TL terms.
Rule \Rule{red-rel-exp} relates FG and TL expressions.
The expressions are in a relation assuming that the resulting values are in a relation
where we impose a step limit on the number of reduction steps that can be taken.
We additionally find
parameters $\foreachN{D}$ and $\vbMethTL$ as FG and TL expressions refer to method definitions.

Rule \Rule{red-rel-struct} is straightforward.
Rule \Rule{red-rel-iface} has been motivated above.
We make use of the following helper function to lookup up the method definition
for a specific pair of method name and receiver type.
\bda{c}\small
  \myirule{ \FUNC\ (x \ t_S) \ mM \ \{ \RETURN\ e \} \in  \foreachN{D} }
          { \methodLookup{\foreachN{D}}{(m,t_S)} = \FUNC\ (x \ t_S) \ mM \ \{ \RETURN\ e \} }
\eda

Rule \Rule{red-rel-method} covers method definitions.
Rule \Rule{red-rel-vb} ensures that the substitutions from free variables to values are related.
Rule \Rule{red-rel-decls} ensures that our labeling for the translation of method definitions is consistent.

A fundamental property of step-indexed logical relations is that
if two expressions are in a relation for $k$ steps then they are also in a relation
for any smaller number of steps.

\begin{lemma}[Monotonicity]
\label{le:monotonicity}
  Let $\redLRConfk{k}{t}{e}{\expT}$ and $k' \leq k$.
  Then, we find that $\redLRConfk{k'}{t}{e}{\expT}$.
\end{lemma}
\begin{proof}
  By induction over the derivation $\redLRConfk{k}{t}{e}{\expT}$.

  \noindent
  {\bf Case} \Rule{red-rel-exp}:
  \begin{mathpar}
 \inferrule
           {  \forall k_1 < k, k_2 < k, v, \uT.
             (k  - k_1 - k_2 > 0 \wedge
             \reduceFGk{k_1}{\foreachN{D}}{e}{v}
             \wedge
             \reduceTLk{k_2}{\vbMethTL}{\expT}{\uT})
             \\
             \implies \redLRConfk{k- k_1 - k_2}{t}{v}{\uT}
 }
 { \redLRConfk{k}{t}{e}{\expT}
 }
  \end{mathpar}

  If either $e$ or $\expT$ is irreducible, $\redLRConfk{k'}{t}{e}{\expT}$
  holds immediately because the universally quantified statement in the premise holds vacuously.

  Otherwise, we find
  $\reduceFGk{k_1}{\foreachN{D}}{e}{v}$
  and
  $\reduceTLk{k_2}{\vbMethTL}{\expT}{\uT}$
  for some $k_1$ and $k_2$. If $k' - k_1 - k_2 \leq 0$, $\redLRConfk{k'}{t}{e}{\expT}$ holds again immediately.

  Otherwise, by induction applied on the premise of rule \Rule{red-rel-exp}
  we find that $\redLRConfk{k'- k_1 - k_2}{t}{v}{\uT}$ and we are done for this case.

  \noindent
      {\bf Case} \Rule{red-rel-struct}:
\begin{mathpar}
 \inferrule
     { \TYPE\ t_S \ \STRUCT\ \{ \foreach{f \ t}{n} \} \in  \foreachN{D}
       \\ \forall i \in [n]. \redLRConfk{k}{t_i}{v_i}{\uT_i}
     }
     { \redLRConfk{k}
       {t_S}
       {t_S \{ \foreach{v}{n} \}}
       {\kT_{t_S} \ (\foreach{\uT}{n})}
     }
 \end{mathpar}

  Follows immediately by induction.

  \noindent
      {\bf Case} \Rule{red-rel-iface}:
\begin{mathpar}
 \inferrule
    {\uT = \kT_{u_S} \ \foreachN{\uT'}
      \\ (1)~ \forall k_1 < k . \redLRConfk{k_1}{u_S}{v}{\uT}
      \\ \methodSpecifications{\foreachN{D}}{t_I} = \{ \foreach{m M}{n} \}
      \\ (2)~ \forall k_2 < k, i \in [n] . \redLRConfk{k_2}
                       {m_i M_i}
                       {\methodLookup{\foreachN{D}}{(m_i,u_S)}}
                       {\uT_i}
     }
     { \redLRConfk{k}{t_I}{v}
             {\kT_{t_I} \ (\uT, \foreach{\uT}{n})}
     }
\end{mathpar}

  Consider the first premise (1).
  If there exists $k_1 < k'$ then $\redLRConfk{k_1}{u_S}{v}{\uT}$. Otherwise,
  this premise holds vacuously.
  The same argument for $k_2 < k'$ applies to the second premise (2).
  Hence,
  $\redLRConfk{k'}{t_I}{v}{\kT_{t_I} \ (\uT, \foreach{\uT_i}{n})}$.
  \qed
\end{proof}
A similar monotonicity result applies to method definitions and declarations.
Monotonicity is an essential property to obtain the following results.

Interface-value constructors and destructors preserve equivalent expressions via logical relations
as stated by the following results.

\begin{lemma}[Structural Subtyping versus Interface-Value Constructors]
\label{le:upcast-red-equiv}
  Let $\tdUpcast{\foreachN{D}}{\subtypeOf{t}{u}}{\expT_1}$
  and $\redLRk{k}{}{}{\foreachN{D}}{\vbMethTL}$
  and $\redLRConfk{k}{t}{e}{\expT_2}$.
  Then, we find that $\redLRConfk{k}{u}{e}{\expT_1 \ \expT_2}$.
\end{lemma}

\begin{lemma}[Type Assertions versus Interface-Value Destructors]
\label{le:downcast-red-equiv}
  Let $\tdDowncast{\foreachN{D}}{\assertOf{t}{u}}{\expT_1}$
  and $\redLRk{k}{}{}{\foreachN{D}}{\vbMethTL}$
  and $\redLRConfk{k}{t}{e}{\expT_2}$.
  Then, we find that $\redLRConfk{k}{u}{e.(u)}{\expT_1 \ \expT_2}$.
\end{lemma}

Based on the above we can show that target expressions resulting from FG expressions
and target methods resulting from FG methods are equivalent.

  \begin{lemma}[Expression Equivalence]
\label{le:exp-red-equivalent}
  Let $\tdExpTrans{\pair{\foreachN{D}}{\fgEnv}}{e : t}{\expT}$
  and $\vbFG$, $\vbTL$, $\vbMethTL$
  such that
  $\redLRk{k}{\triple{\foreachN{D}}{\vbMethTL}{\fgEnv}}{}{\vbFG}{\vbTL}$
  and
  $\redLRk{k}{}{}{\foreachN{D}}{\vbMethTL}$ for some $k$.
  Then, we find that $\redLRConfk{k}{t}{\vbFG(e)}{\vbTL(\expT)}$.
\end{lemma}

  \begin{lemma}[Method Equivalence]
    \label{le:method-red-rel-equiv}
 \mbox{} \\ 
  Let $\tdMethTrans{\foreachN{D}}
                        {\FUNC\ (x \ t_S) \ m (\foreach{x \ t}{n}) \ t \ \{ \RETURN\ e \}}
                        {\lambda \xT . \lambda (\foreach{\xT }{n}). \expT}$.
 Then, we find that $\redLRk{k}{}{}{\foreachN{D}}{\vbMethTL}$ where $\vbMethTL(\mT{m}{t_S}) = \lambda \xT . \lambda (\foreach{\xT }{n}). \expT$ for any $k$.
  \end{lemma}

  The lengthy proofs of the above results are given in the appendix.

  From Lemmas~\ref{le:exp-red-equivalent} and~\ref{le:method-red-rel-equiv}
  we can derive our main result that the dictionary-passing translation preserves the dynamic behavior of FG programs.

\begin{theorem}[Program Equivalence]
\label{theo:prog-equiv}
Let $\tdProgTrans{\ignore{\PACKAGE\ \MAIN;} \foreachN{D} \ \FUNC\ \MAIN () \{ \_ = e \}}
{\LET\ \foreach{\mT{m_i}{{t_S}_i} = \expT_i}{n} \ \IN\ \expT}$
where we assume that $e$ has type $t$.
Then, we find that
$\redLRConfk{k}{t}{e}{\expT}$
for any $k$
where $\vbMethTL = \Angle{\foreach{\subst{\mT{m_i}{{t_S}_i}}{\expT_i}}{n}}$.
\end{theorem}
\begin{proof}
  Follows from Lemmas~\ref{le:exp-red-equivalent} and~\ref{le:method-red-rel-equiv}.
  \qed
\end{proof}

Our main result also implies that our translation is coherent.
Recall that the translation rules are non-syntax directed because of rule  \Rule{td-sub}.
Hence, we could for example insert an (albeit trivial) interface-value constructor resulting from
$\tdUpcast{\foreachN{D}}{\subtypeOf{t_I}{t_I}}{\expT}$.
Hence, there might be different target terms for the same source term.
Our main result guarantees that all targets obtained
preserve the meaning of the original program.

\section{Related Work and Conclusion}
\label{sec:related-conc}

The dictionary-passing translation is well-studied in the context of
Haskell type classes~\cite{Wadler:1989:MAP:75277.75283}. A type class
constraint translates to an extra function parameter, constraint resolution
provides a dictionary with the methods of the type class for this parameter.
In our translation from Featherweight Go~\cite{FeatherweightGo}, dictionaries are not
supplied as separate parameters because FG does
not support parametric polymorphism. Instead, a dictionary is always passed as part
of an interface-value, which combines the dictionary with the concrete
value implementing the interface. Thus, interface-values can be viewed as
representations of
existential types~\cite{DBLP:journals/toplas/MitchellP88,laeufer_1996,ThiemannWehr2008}.
How to adapt our dictionary-passing translation scheme to FG extended with parametric polymorphism (generics)
is something we plan to consider in future work.

In the context of type classes it is common to show
that resulting target programs are well-typed.
For example, see the work by Hall and coworkers~\cite{Hall:1996:TCH:227699.227700}.
Typed target terms in this setting require System F
and richer variants depending on the kind of type class extensions
that are considered~\cite{10.1145/1190315.1190324}.
Our target terms are untyped and
we pattern match over constructors to check for ``run-time types''.
For example, see rule \Rule{td-destr-iface-struct} in Figure~\ref{f:upcast-downcast}.
There are various ways to support dynamic typing in a typed setting.
For example, we could employ GADTs as described by
Peyton Jones and coworkers~\cite{DBLP:conf/birthday/JonesWEV16}.
A simply-typed first order functional language with GADTs appears then to be sufficient
as a typed target language for Featherweight Go.
This will require certain adjustments to our dictionary-passing translation.
We plan to study the details in future work.

Another important property in the type class context is coherence.
Bottu and coworkers~\cite{10.1145/3341695} make use of logical relations
to state equivalence among distinct target terms resulting
from the same source type class program.
Thanks to our main result Theorem~\ref{theo:prog-equiv}, we get coherence for free.
We believe it is worthwhile to establish a property similar to Theorem~\ref{theo:prog-equiv} for type classes.
We could employ a simple denotational semantics
for source type class programs
such as~\cite{DBLP:conf/lfp/Thatte94,10.1145/2633357.2633364}
which is then related to target programs obtained via the dictionary-passing translation.
This is something that has not been studied so far
and another topic for future work.

Method dictionaries bear some resemblance to virtual method tables
(vtables) used to implement virtual method dispatch in
object-oriented languages~\cite{DBLP:conf/oopsla/DriesenH96}.
The main difference between vtables and
dictionaries is that there is a fixed connection between an object and its
vtable (via the class of the object), whereas the connection between a
value and a dictionary may change at runtime, depending on the type
the value is used at. Dictionaries allow access to a method at a fixed
offset, whereas vtables in the presence of multiple inheritance require
a more sophisticated lookup algorithm~\cite{DBLP:conf/oopsla/AlpernCFGL01}.

Subtyping for interfaces in Go is based purely on width subtyping,
there is no support for depth subtyping~\cite{TAPL}: a subtype
might provide more methods than the super-interface, but method signatures
must match invariantly. Method dispatch in Go is performed only
on the receiver of the method call. Multi-dispatch~\cite{CLOS,DBLP:conf/ecoop/Chambers92}
refers to the ability to dispatch on multiple arguments, but this approach
turns out to be difficult in combination with structural subtyping~\cite{DBLP:conf/ecoop/MalayeriA08}.

To summarize the results of the paper at hand: we defined a dictionary-passing
translation from Featherweight Go to a untyped lambda calculus with pattern matching.
The compiler for the full Go language~\cite{golang} employs a similar
dictionary-passing approach. We proved that the translation preserves
the dynamic semantics of Featherweight Go, using step-indexed logical
relations.



\bibliography{main}


\appendix

\section{Proofs for Properties Stated in the Main Text}

\subsection{Monotonicity for Method Definitions and Declarations}

\begin{lemma}[Monotonicity 2]
\label{le:monotonicity2}
Let $\redLRConfk{k}{m M}{\FUNC\ (x \ t_S) \ mM \ \{ \RETURN\ e \}}{\uT}$ and $k' \leq k$.
Then, we find that $\redLRConfk{k'}{m M}{\FUNC\ (x \ t_S) \ mM \ \{ \RETURN\ e \}}{\uT}$.
\end{lemma}
\begin{proof}
  Follows immediately by observing the premise of rule \Rule{red-rel-method}.
  \qed
\end{proof}

\begin{lemma}[Monotonicity 3]
\label{le:monotonicity3}
Let $\redLRk{k}{}{}{\foreachN{D}}{\vbMethTL}$ and $k' \leq k$.
  Then, we find that $\redLRk{k'}{}{}{\foreachN{D}}{\vbMethTL}$.
\end{lemma}
\begin{proof}
  Follows via Lemma~\ref{le:monotonicity2}.
  \qed
\end{proof}

\subsection{Lemma~\ref{le:upcast-red-equiv}}

\begin{proof}

  We show that $\redLRConfk{k}{u}{e}{\expT_1 \ \expT_2}$ by
  making use of the following auxiliary statement.

  Let $\tdUpcast{\foreachN{D}}{\subtypeOf{t}{u}}{\expT}$
  and $\redLRk{k}{}{}{\foreachN{D}}{\vbMethTL}$
  and $\redLRConfk{k}{t}{v}{\uT}$.
  Then, we find that $\redLRConfk{k}{u}{v}{\expT \ \uT}$.

  Suppose $k_1 < k$ and $k_2 < k$ and $k - k_1 - k_2 > 0$
  and (1) $\reduceFGk{k_1}{\foreachN{D}}{e}{v}$ for some $v$
  and (2) $\reduceTLk{k_2}{\vbMethTL}{\expT_2}{\uT}$ for some $\uT$.
  Based on the assumption that $\redLRConfk{k}{t}{e}{\expT_2}$
  and via rule \Rule{red-rel-exp} we conclude that
  $\redLRConfk{k - k_1 - k_2}{t}{v}{\uT}$.

  Via the auxiliary statement we conclude that
  (3) $\redLRConfk{k - k_1 - k_2}{u}{v}{\expT_1 \ \uT}$.

  Via rule \Rule{red-rel-exp} making use of (1), (2) and (3)
  we conclude that $\redLRConfk{k}{u}{e}{\expT_1 \ \expT_2}$ and we are done.

Proof of auxiliary statement.

We have to show that for all $k_2 < k$
where $k - k_2 > 0$ and $\reduceTLk{k_2}{\vbMethTL}{\expT \ \uT}{\uT'}$
we have that $\redLRConfk{k-k_2}{t}{v}{\uT'}$.

We perform a case analysis of the derivation for $\tdUpcast{\foreachN{D}}{\subtypeOf{t}{u}}{\expT}$
and label the assumptions
(1) $\redLRk{k}{}{}{\foreachN{D}}{\vbMethTL}$
and (2) $\redLRConfk{k}{t}{v}{\uT}$ for later reference.

  \noindent
      {\bf Case} \Rule{td-cons-struct-iface}:
\begin{mathpar}
\inferrule
          {\TYPE\ t_I \ \INTERFACE\ \{ \foreachN{S} \} \in \foreachN{D}
        \\ \methodSpecifications{\foreachN{D}}{t_S} \supseteq \foreachN{S}
        \\ \foreachN{S} = \foreach{m M}{n}
          }
          {\tdUpcast{\foreachN{D}}{\subtypeOf{t_S}{t_I}}
                    {\lambda \xT.
                      \kT_{t_I} \ (\xT, \foreach{\mT{m_i}{t_S}}{n}) }
          }
\end{mathpar}

Set $\expT = \lambda \xT. \kT_{t_I} \ (\xT, \foreach{\mT{m_i}{t_S}}{n})$.
Then, (3) $\reduceTLk{1}{\vbMethTL}{\expT \ \uT}{\uT'}$
where $\uT' = \kT_{t_I} \ (\uT, \foreach{\mT{m_i}{t_S}}{n})$.

From (1) and Lemma~\ref{le:monotonicity3} we obtain
(4) $\forall k_1 < k. \redLRk{k_1}{}{}{\foreachN{D}}{\vbMethTL}$.

From (2) and Lemma~\ref{le:monotonicity} we obtain
(5) $\forall k_2 < k. \redLRConfk{k_2}{t_S}{v}{\uT}$ (for this case $t = t_S$).

From (4), (5) and via rule \Rule{red-rel-iface} we obtain
(6) $\redLRConfk{k}{t_I}{v}{\uT'}$.

From (3), (6) and via rule \Rule{red-rel-exp} we obtain
$\redLRConfk{k}{t_I}{v}{\expT \ \uT}$ and we are done for this case.

\noindent
{\bf Case} \Rule{td-cons-iface-iface}:
\begin{mathpar}
    \inferrule
          {\TYPE\ t_I \ \INTERFACE\ \{ \foreach{R_i}{n} \} \in \foreachN{D}
       \\ \TYPE\ u_I \ \INTERFACE\ \{ \foreach{S_i}{q} \} \in \foreachN{D}
       \\  S_i = R_{\mapPerm(i)} \quad\noteForall{i \in [q]}
        }
          {\tdUpcast{\foreachN{D}}{\subtypeOf{t_I}{u_I}}
                    { \lambda \xT. \CASE\,\xT\,\OF\
                          \kT_{t_I} \ (\xTval, \foreach{\xT_i}{n})
                          \rightarrow
                         \kT_{u_I} \ (\xTval, \foreach{\xT_{\mapPerm(i)}}{q})
                    }
         }
\end{mathpar}

From (2) and for this case we can conclude via rule \Rule{red-rel-iface}
that (3) $\uT = \kT_{t_I} \ (\uT', \foreach{\uT_i}{n})$
for some $u_S$, $\uT'$ and $\foreach{\uT_i}{n}$
where
$\methodSpecifications{\foreachN{D}}{t_I} = \{ \foreach{m_i M_i}{n} \}$
and
(4) $\forall k_1 < k . \redLRConfk{k_1}{u_S}{v}{\uT'}$
and
(5) $\forall k_2 < k, i \in [n] . \redLRConfk{k_2}{m_i M_i}{\methodLookup{\foreachN{D}}{(m_i,u_S)}}{\uT_i}$.

Set $\expT = \lambda \xT. \CASE\,\xT\,\OF\
                          \kT_{t_I} \ (\xTval, \foreach{\xT_i}{n})
                          \rightarrow
                          \kT_{u_I} \ (\xTval, \foreach{\xT_{\mapPerm(i)}}{q})$.
Then, (6) $\reduceTLk{1}{\vbMethTL}{\expT \ \uT}{\uT''}$
where (7) $\uT'' = \kT_{u_I} \ (\uT', \foreach{\uT_{\mapPerm(i)}}{q})$.

From (4), (5) and (7) via rule \Rule{red-rel-iface} we obtain that
(8) $\redLRConfk{k}{u_I}{v}{\uT''}$.

From (6), (8) and via rule \Rule{red-rel-exp} we obtain
$\redLRConfk{k}{u_I}{v}{\expT \ \uT}$ and we are done.
\qed
\end{proof}

\subsection{Lemma~\ref{le:downcast-red-equiv}}

We first introduce an auxiliary statement.

\begin{lemma}
\label{le:k-steps}
  Let $\redLRConfk{k}{t}{e'}{\expT'}$
  and $\reduceFGk{k_1}{\foreachN{D}}{e}{e'}$
  and $\reduceTLk{k_2}{\vbMethTL}{\expT}{\expT'}$.
  Then, we find that $\redLRConfk{k + k_1 + k_2}{t}{e}{\expT}$
\end{lemma}
\begin{proof}
  If either $e'$ or $\expT'$ are irreducible the result follows immediately.

  Otherwise, based on rule \Rule{red-rel-exp} we find that
  $\reduceFGk{k_1'}{\foreachN{D}}{e'}{v}$
  and $\reduceTLk{k_2'}{\vbMethTL}{\expT'}{\uT}$
  and $\redLRConfk{k-k_1' - k_2'}{t}{v}{\uT}$.

  Based on the above, our assumptions and rule \Rule{red-rel-exp}
  we find that $\redLRConfk{k + k_1 + k_2}{t}{e}{\expT}$ and we are done.
  \qed
\end{proof}

Here comes the proof of Lemma~\ref{le:downcast-red-equiv}.

\begin{proof}
  We perform a case analysis of the derivation $\tdDowncast{\foreachN{D}}{\assertOf{t}{u}}{\expT_1}$
  and label the assumptions
(1) $\redLRk{k}{}{}{\foreachN{D}}{\vbMethTL}$
and (2) $\redLRConfk{k}{t}{e}{\expT_2}$ for later reference.

  \noindent
  {\bf Case} \Rule{td-destr-iface-struct}:
  \begin{mathpar}
  \inferrule
            {\TYPE\ t_I \ \INTERFACE\ \{ \foreach{S}{n} \} \in \foreachN{D}
             \\ \fgSub{\foreachN{D}}{t_S}{t_I}
            }
            {\tdDowncast{\foreachN{D}}{\assertOf{t_I}{t_S}}{\lambda \xT.  \CASE\,\xT\,\OF\
                    \kT_{t_I} \ (K_{t_S} \ \foreachN{\yT}, \foreach{\xT}{n})
                      \rightarrow K_{t_S} \ \foreachN{\yT}
                   }
            }
  \end{mathpar}

  We set $\expT_1  = \lambda \xT.  \CASE\,\xT\,\OF\
                    \kT_{t_I} \ (K_{t_S} \ \foreachN{\yT}, \foreach{\xT}{n})
                      \rightarrow K_{t_S} \ \foreachN{\yT}$.

  From (2) and via rule \Rule{red-rel-exp} we conclude that
  forall  $k_1 < k$, $k_2 < k$, $v$, $\uT$ where
  (3) $k  - k_1 - k_2 > 0$
  and
  (4) $\reduceFGk{k_1}{\foreachN{D}}{e}{v}$
  and
  (5) $\reduceTLk{k_2}{\vbMethTL}{\expT_2}{\uT}$
  we have that
  (6) $\redLRConfk{k- k_1 - k_2}{t_I}{v}{\uT}$.

  From (6) and rule \Rule{red-rel-iface} we conclude that
  (7) $\uT = \kT_{t_I} \ (\RepK{u_S}, \uT, \foreach{\uT_i}{n})$
  where $\uT = \kT_{u_S} \ \foreachN{\uT'}$ and for all (8) $k_1' < k- k_1 - k_2$ we have that
  (9) $\redLRConfk{k_1'}{u_s}{v}{\uT}$.

  \noindent
  {\bf Subcase $t_S \not = u_S$:}
  Neither $e.(t_S)$ nor $\expT_1 \ \expT_2$ are reducible and
  therefore we immediately can conclude that
  $\redLRConfk{k}{t_S}{e}{\expT_1 \ \expT_2}$ holds.

    \noindent
   {\bf Subcase $t_S = u_S$:}
   From (4) we conclude that
   (10) $\reduceFGk{k_1 + 1}{\foreachN{D}}{e.(t_S)}{v}$.

   From (5) and (7) we conclude that
   (11) $\reduceTLk{k_2 + 3}{\vbMethTL}{\expT_1 \ \expT_2}{\uT}$.
   There are three additional reduction steps
   as we have one extra lambda and two extra pattern match applications.

   From (8) and (9) and the Monotonicity Lemma~\ref{le:monotonicity} we conclude that
   $\redLRConfk{k- (k_1+1) - (k_2 + 3)}{t_s}{v}{\uT}$ and via rule \Rule{red-rel-exp}
   we obtain that $\redLRConfk{k}{t_S}{e}{\expT_1 \ \expT_2}$ and we are done for this case.

   \noindent
   {\bf Case} \Rule{td-destr-iface-iface}:
   \begin{mathpar}
    \inferrule
          { \TYPE\ t_I \ \INTERFACE\ \{ \foreach{S}{n} \} \in \foreachN{D}
            \\
            \foreachN{\clsT} =
            [ \kT_{t_S} \ \foreachN{\yT} \rightarrow (\expT \ \xTval)
               \mid \tdUpcast{\foreachN{D}}{\subtypeOf{t_S}{u_I}}{\expT} ]
          }
          { \tdDowncast{\foreachN{D}}{\assertOf{t_I}{u_I}}{\lambda \xT.
                             \CASE\ \ \xT \ \OF\
                               \kT_{t_I} \ (\xTval, \foreach{\xT}{n}) \rightarrow \CASE\ \xTrep \ \OF\ \foreachN{\clsT}}
          }
   \end{mathpar}

   We set $\expT_1 = \lambda \xT.
                             \CASE\ \ \xT \ \OF\
                             \kT_{t_I} \ (\xTval, \foreach{\xT}{n}) \rightarrow \CASE\ \xTrep \ \OF\ \foreachN{\clsT}$.

   We apply similar reasoning as in case of \Rule{td-destr-iface-struct}.

    From (2) and via rule \Rule{red-rel-exp} we conclude that
  forall  $k_1 < k$, $k_2 < k$, $v$, $\uT$ where
  (3) $k  - k_1 - k_2 > 0$
  and
  (4) $\reduceFGk{k_1}{\foreachN{D}}{e}{v}$
  and
  (5) $\reduceTLk{k_2}{\vbMethTL}{\expT_2}{\uT}$
  we have that
  (6) $\redLRConfk{k- k_1 - k_2}{t_I}{v}{\uT}$.

  From (6) and rule \Rule{red-rel-iface} we conclude that
  (7) $\uT = \kT_{t_I} \ (\uT, \foreach{\uT_i}{n})$
  where $\uT = \kT_{t_S} \ \foreachN{\uT'}$ and for all (8) $k_1' < k- k_1 - k_2$ we have that
  (9) $\redLRConfk{k_1'}{t_s}{v}{\uT}$.
  We use here $t_S$ (instead of $u_S$) to match the naming conventions
  in the premise of rule \Rule{td-destr-iface-iface}.

    \noindent
  {\bf Subcase $\fgSub{\foreachN{D}}{t_S}{u_I}$ does not hold:}
    Neither $e.(t_S)$ nor $\expT_1 \ \expT_2$ are reducible and
  therefore we immediately can conclude that
  $\redLRConfk{k}{t_S}{e}{\expT_1 \ \expT_2}$ holds.

    \noindent
  {\bf Subcase $\fgSub{\foreachN{D}}{t_S}{u_I}$ does hold:}
      From (4) we conclude that
   (10) $\reduceFGk{k_1}{\foreachN{D}}{e.(u_I)}{v}$.

      From (5) and (7) we conclude that
      (11) $\reduceTLk{k_2 + 3}{\vbMethTL}{\expT_1 \ \expT_2}{\expT \ \uT}$
      where (12) $\tdUpcast{\foreachN{D}}{\subtypeOf{t_S}{u_I}}{\expT}$.
      There are three additional reduction steps
      as we have one extra lambda and two extra pattern match applications.
      The upcast $\expT$ has not been applied.

      From (6) and (12) and Lemma~\ref{le:upcast-red-equiv} we obtain that
      (13) $\redLRConfk{k_1'}{u_I}{v}{\expT \ \uT}$.

      Via the Monotonicity Lemma~\ref{le:monotonicity} and Lemma~\ref{le:k-steps} we obtain that
      $\redLRConfk{k}{u_I}{e.(u_I)}{\expT_1 \ \expT_2}$ and we are done.
      \qed
\end{proof}

\subsection{Lemma~\ref{le:exp-red-equivalent}}

\begin{proof}

  By induction over the derivation $\tdExpTrans{\pair{\foreachN{D}}{\fgEnv}}{e : t}{\expT}$.
  We label the assumptions   (1) $\redLRk{k}{\triple{\foreachN{D}}{\vbMethTL}{\fgEnv}}{}{\vbFG}{\vbTL}$
  and
  (2) $\redLRk{k}{}{}{\foreachN{D}}{\vbMethTL}$ as well as the to be proven
  statement (3) $\redLRConfk{k}{t}{\vbFG(e)}{\vbTL(\expT)}$ for some later reference.

  \noindent
  {\bf Case} \Rule{td-var}:
  \begin{mathpar}
    \inferrule
            { (x : t) \in \fgEnv
            }
            { \tdExpTrans{\pair{\foreachN{D}}{\fgEnv}}{x : t}{\xT}
            }
  \end{mathpar}

  (3) follows immediately from (1).

  \noindent
      {\bf Case} \Rule{td-struct}:
  \begin{mathpar}
    \inferrule
        {   \TYPE\ t_S \ \STRUCT\ \{ \foreachN{f_i \ t_i}{n} \} \in  \foreachN{D}
               \\   \tdExpTrans{\pair{\foreachN{D}}{\fgEnv}}{e_i : t_i}{\expT_i}\quad\noteForall{i \in [n]}
            }
            { \tdExpTrans{\pair{\foreachN{D}}{\fgEnv}}{t_S \{ \foreach{e_i}{n} \} : t_S }{\kT_{t_S} \ (\foreach{\expT_i}{n})}
            }
  \end{mathpar}

  Suppose there exists $k_1 < k$ and $k_2 < k$ and $v$ and $\uT$
  such that
  $(k  - k_1 - k_2 > 0$
  and
  (4) $\reduceFGk{k_1}{\foreachN{D}}{\vbFG(t_s \{ \foreach{e_i}{n} \})}{t_s \{ \foreach{v_i}{n} \}}$
  and
  (5) $\reduceTLk{k_2}{\vbMethTL}{\vbTL(\kT_{t_S} \ (\foreach{\expT_i}{n}))}{\kT_{t_S} \ (\foreach{\uT_i}{n})}$
  for $i \in [n]$.

  From (4) and (5) we conclude that
  (6) $\reduceFGk{k_1'}{\foreachN{D}}{\vbFG(e_i)}{v_i}$
  and
  (7) $\reduceTLk{k_2'}{\vbMethTL}{\vbTL(\expT_i)}{\uT_i}$
  for $i \in [n]$ where we pick $k_1'$ and $k_2$' such that $k_1' < k_1$ and $k_2' < k_2$
  and all the subreductions yield some value.

  By induction (8) $\redLRConfk{k}{t_i}{\vbFG(e_i)}{\vbTL(\expT_i)}$ for $i \in [n]$.

  From (6), (7), (8) and via rule \Rule{red-rel-exp} we conclude that
  (9) $\redLRConfk{k-k_1'-k_2'}{t_i}{v_i}{\uT_i}$ for $i \in [n]$.

  From (9) and rule \Rule{red-rel-struct} we conclude that
  (10) $\redLRConfk{k-k_1'-k_2'}{t_S}{t_S \{ \foreach{v_i}{n} \}}{\kT_{t_S} \ (\foreach{\uT_i}{n})}$.

  From (10) and Lemma~\ref{le:monotonicity} we conclude that
  (11) $\redLRConfk{k-k_1-k_2}{t_S}{t_S \{ \foreach{v_i}{n} \}}{\kT_{t_S} \ (\foreach{\uT_i}{n})}$.

  From (4), (5), (11) and via rule \Rule{red-rel-exp} we conclude that
  $\redLRConfk{k}{t_S}{\vbFG(t_s \{ \foreach{e_i}{n} \})}{\vbTL(\kT_{t_S} \ (\foreach{E_i}{n}))}$
  and we are done for this case.

  \noindent
  {\bf Case} \Rule{td-access}:
  \begin{mathpar}
      \inferrule
           {  \tdExpTrans{\pair{\foreachN{D}}{\fgEnv}}{e : t_S}{\expT}
             \\  \TYPE\ t_S \ \STRUCT\ \{ \foreach{f_j \ t_j}{n} \} \in  \foreachN{D}
          }
           { \tdExpTrans{\pair{\foreachN{D}}{\fgEnv}}
                        {e.f_i : t_i }
                        { \CASE\ \ \expT\ \ \OF\ \kT_{t_S} \ (\foreach{\xT_j}{n}) \rightarrow \xT_i}
           }
  \end{mathpar}

  Similar reasoning as in case of \Rule{td-struct}.

    \noindent
  {\bf Case} \Rule{td-call-struct}:
  \begin{mathpar}
      \inferrule
            { m (\foreach{x_i \ t_i}{n}) \ t \in \methodSpecifications{\foreachN{D}}{t_S}
              \\
              \tdExpTrans{\pair{\foreachN{D}}{\fgEnv}}{e : t_S}{\expT}
              \\ \tdExpTrans{\pair{\foreachN{D}}{\fgEnv}}{e_i : t_i}{\expT_i}\quad\noteForall{i \in [n]}
          }
            { \tdExpTrans{\pair{\foreachN{D}}{\fgEnv}}{e.m(\foreach{e_i}{n}) : t}{\mT{m}{t_S} \ \expT \ (\foreach{\expT_i}{n}) } }
  \end{mathpar}

    Suppose there exists $k_1 < k$ and $k_2 < k$ and $v$ and $\uT$
  such that
  $(k  - k_1 - k_2 > 0$
  and
  (4) $\reduceFGk{k_1}{\foreachN{D}}{\vbFG(e.m(\foreach{e_i}{n}))}{v}$
  and
  (5) $\reduceTLk{k_2}{\vbMethTL}{\vbTL(\mT{m}{t_S} \ \expT \ (\foreach{\expT_i}{n}))}{\uT}$.

  From the assumptions and (4) we conclude that
  (4a) $\reduceFGk{1}{\foreachN{D}}{\vbFG(e.m(\foreach{e_i}{n}))}{\Angle{\subst{x}{\vbFG(e)},\foreach{\subst{x_i}{\vbFG(e_i)}}{n}} e'}$
  and
  (4b) $\reduceFGk{k_1 - 1}{\foreachN{D}}{\Angle{\subst{x}{\vbFG(e)},\foreach{\subst{x_i}{\vbFG(e_i)}}{n}} e'}{v}$
  where (4c) $\FUNC\ (x \ t_S) \ m (\foreach{x_i \ t_i}{n}) \ t \ \{ \RETURN\ e' \} \in \foreachN{D}$.

  From (4) we conclude that
  (6) $\reduceFGk{k_1'}{\foreachN{D}}{\vbFG(e)}{v'}$
  and
  (7) $\reduceFGk{k_1'}{\foreachN{D}}{\vbFG(e_i)}{v_i}$
  for some $v'$ and $v_i$ for $i \in [n]$ where $k_1' < k_1$.
  We pick again some large enough $k_1'$ such that all subreductions yields some value.

  Similarly, from (5) we conclude that
  (8) $\reduceTLk{k_2'}{\vbMethTL}{\vbTL(\expT)}{\uT'}$
  and
  (9) $\reduceTLk{k_2'}{\vbMethTL}{\vbTL(\expT_i)}{\uT_i}$
  for some $\uT'$ and $\uT_i$ for $i \in [n]$ where $k_2' < k_2$.

  By induction we have that
  (10) $\redLRConfk{k}{t_S}{\vbFG(e)}{\vbTL(\expT)}$
  and
  (11) $\redLRConfk{k}{t_i}{\vbFG(e_i)}{\vbTL(\expT_i)}$ for $i \in [n]$.

  From (6), (8), (10) and via rule \Rule{red-rel-exp} we conclude that
  (12) $\redLRConfk{k-k_1'-k_2'}{t_S}{v'}{\uT'}$.

  Similarly, from (7), (9), (11) and via rule \Rule{red-rel-exp} we conclude that
  (13) $\redLRConfk{k-k_1'-k_2'}{t_i}{v_i}{\uT_i}$ for $i \in [n]$.

  From (4c), (12), (13), (3) and via rule \Rule{red-rel-method} we conclude that
  (14) $\redLRConfk{k-k_1'-k_2'}{t}{\Angle{\subst{x}{v},\foreach{\subst{x_i}{v_i}}{n}} e'}{\mT{m}{t_S} \ \uT' \ (\foreach{\uT_i}{n})}$.

  Based on our choice of $k_1'$ and $k_2'$ we conclude that
  (15) $\reduceFGk{k_1 - k_1'}{\foreachN{D}}{\Angle{\subst{x}{v},\foreach{\subst{x_i}{v_i}}{n}} e'}{v}$
  and
  (16) $\reduceTLk{k_2 - k_2'}{\vbMethTL}{\mT{m}{t_S} \ \uT' \ (\foreach{\uT_i}{n})}{\uT}$.
  That is, with $k_1 - k_1'$ steps or less we reach $v$ because $k_1$ is the overall number of steps required
  and $k_1'$ is the maximum number of one of the subcomputation steps.
  The same applies to $k_2 - k_2'$.

  From (14), (15), (16) and via rule \Rule{red-rel-exp} we conclude that
  $\redLRConfk{k - k_1 - k_2}{t}{v}{\uT}$ where we make use of the fact that
  $k - k_1' - k_1' - (k_1 - k_1') - (k_2 - k_2') = k - k_1 - k_2$.
  Thus, we are done for this case.

  \noindent
  {\bf Case} \Rule{td-call-iface}:
  \begin{mathpar}
  \inferrule
            { \tdExpTrans{\pair{\foreachN{D}}{\fgEnv}}{e : t_I}{\expT}
             \\ \TYPE\ t_I \ \INTERFACE\ \{ \foreach{S_i}{q} \} \in \foreachN{D}
             \\ S_j = m(\foreach{x_i \ t_i}{n}) \ t\quad(\textrm{for some}~j \in [q])
            \\ \tdExpTrans{\pair{\foreachN{D}}{\fgEnv}}{e_i : t_i}{\expT_i}\quad\noteForall{i \in [n]}
          }
          { \tdExpTrans{\pair{\foreachN{D}}{\fgEnv}}
              {e.m(\foreach{e_i}{n}) : t}
              {\CASE\ \ \expT \ \OF\ \kT_{t_I} \ (\xTval, \foreach{\xT_i}{q}) \rightarrow \xT_j \ \xTval \ (\foreach{\expT_i}{n}) }
          }
  \end{mathpar}

  Similar reasoning as in case of \Rule{td-call-struct}.
  We set
  $$\expT' = \CASE\ \ \expT \ \OF\ \kT_{t_I} \ (\_, \xTval, \foreach{\xT_i}{q}) \rightarrow \xT_j \ \xTval \ (\foreach{\expT_i}{n}).$$

  Suppose there exists $k_1 < k$ and $k_2 < k$ and $v$ and $\uT$
  such that
  $(k  - k_1 - k_2 > 0$
  and
  (4) $\reduceFGk{k_1}{\foreachN{D}}{\vbFG(e.m(\foreach{e_i}{n}))}{v}$
  and
  (5) $\reduceTLk{k_2}{\vbMethTL}{\vbTL(\expT')}{\uT}$.

  From the assumptions and (4) we conclude that
  (4a) $\reduceFGk{1}{\foreachN{D}}{\vbFG(e.m(\foreach{e_i}{n}))}{\Angle{\subst{x}{\vbFG(e)},\foreach{\subst{x_i}{\vbFG(e_i)}}{n}} e'}$
  and
  (4b) $\reduceFGk{k_1 - 1}{\foreachN{D}}{\Angle{\subst{x}{\vbFG(e)},\foreach{\subst{x_i}{\vbFG(e_i)}}{n}} e'}{v}$
  where (4c) $\FUNC\ (x \ t_S) \ m (\foreach{x_i \ t_i}{n}) \ t \ \{ \RETURN\ e' \} \in \foreachN{D}$.

  From (4) we conclude that
  (6) $\reduceFGk{k_1'}{\foreachN{D}}{\vbFG(e)}{v'}$
  and
  (7) $\reduceFGk{k_1'}{\foreachN{D}}{\vbFG(e_i)}{v_i}$
  for some $v'$ and $v_i$ for $i \in [n]$ where $k_1' < k_1$.
  We pick again some large enough $k_1'$ such that all subreductions yields some value.

  Similarly, from (5) we conclude that
  (8) $\reduceTLk{k_2'}{\vbMethTL}{\vbTL(\expT)}{\uT'}$
  and
  (9) $\reduceTLk{k_2'}{\vbMethTL}{\vbTL(\expT_i)}{\uT_i}$
  for some $\uT'$ and $\uT_i$ for $i \in [n]$ where $k_2' < k_2$.

  By induction we have that
  (10) $\redLRConfk{k}{t_S}{\vbFG(e)}{\vbTL(\expT)}$
  and
  (11) $\redLRConfk{k}{t_i}{\vbFG(e_i)}{\vbTL(\expT_i)}$ for $i \in [n]$.

  From (6), (8), (10) and via rule \Rule{red-rel-exp} we conclude that
  (12) $\redLRConfk{k-k_1'-k_2'}{t_I}{v'}{\uT'}$.

  Similarly, from (7), (9), (11) and via rule \Rule{red-rel-exp} we conclude that
  (13) $\redLRConfk{k-k_1'-k_2'}{t_i}{v_i}{\uT_i}$ for $i \in [n]$.

  From (12) and via \Rule{red-rel-iface} we conclude that
  (13) $\uT' = \kT_{t_I} \ (\uT'', \foreach{\uT'_i}{p})$
  and
  $\uT'' = \kT_{t_S} \ \foreachN{\uT'''}$
  and
  (14) $\redLRConfk{k''}{t_S}{v'}{\uT''}$
  and
  (15) $\redLRConfk{k''}{m(\foreach{x_i \ t_i}{n})}{\FUNC\ (x \ t_S) \ m (\foreach{x_i \ t_i}{n}) \ t \ \{ \RETURN\ e' \}}{\uT'_j}$ for some $k''$
  where $k'' < k - k_1' - k_2'$ and $j \in [q]$ is the same $j$ as in the premise of rule \Rule{td-call-iface}.

  From (15) via rule \Rule{red-rel-method} and (14) and (13) plus the Monotonicity Lemma~\ref{le:monotonicity}
  we conclude that
  (16) $\redLRConfk{k''}{t}{\Angle{\subst{x}{v'},\foreach{\subst{x_i}{v_i}}{n}} e'}{\uT'_j \ \uT'' \ (\foreach{\uT_i}{n})}$.

  For concreteness, we can assume $k'' = k - k_1' - k_2' - 1$.
  Based on our choice of $k_1'$ and $k_2'$ we conclude that
  (17) $\reduceFGk{k_1 - k_1'}{\foreachN{D}}{\Angle{\subst{x}{v},\foreach{\subst{x_i}{v_i}}{n}} e'}{v}$
  and
  (18) $\reduceTLk{k_2 - k_2' + 1}{\uT'_j \ \uT'' \ (\foreach{\uT_i}{n})}{\uT}$.
  The argument is the same as in case of \Rule{td-call-struct}.

  From (16), (17), (18) and via rule \Rule{red-rel-exp} we conclude that
  $\redLRConfk{k - k_1 - k_2}{t}{v}{\uT}$ and we are done for this case.

  \noindent
  {\bf Case} \Rule{td-sub}:
  \begin{mathpar}
  \inferrule
            {\tdExpTrans{\pair{\foreachN{D}}{\fgEnv}}{e : t}{\expT_2}
          \\ \tdUpcast{\foreachN{D}}{\subtypeOf{t}{u}}{\expT_1}
          }
            { \tdExpTrans{\pair{\foreachN{D}}{\fgEnv}}{e : u}{\expT_1 \ \expT_2} }
  \end{mathpar}

  By induction we obtain that
  (4) $\redLRConfk{k}{t}{\vbFG(e)}{\vbTL(\expT_2)}$.
  From (3), (4) and Lemma~\ref{le:upcast-red-equiv}  we obtain that
  $\redLRConfk{k}{u}{\vbFG(e)}{\expT_1 \ \vbTL(\expT_2)}$.

  We have that $\vbTL(\expT_1) = \expT_1$ and thus we are done for this case.

  \noindent
  {\bf Case} \Rule{td-assert}:
  \begin{mathpar}
  \inferrule
            {\tdExpTrans{\pair{\foreachN{D}}{\fgEnv}}{e : u}{\expT_2}
          \\ \tdDowncast{\foreachN{D}}{\assertOf{u}{t}}{\expT_1}
          }
          { \tdExpTrans{\pair{\foreachN{D}}{\fgEnv}}{e.(t) : t}{\expT_1 \ \expT_2} }
  \end{mathpar}

    By induction we obtain that
  (4) $\redLRConfk{k}{u}{\vbFG(e)}{\vbTL(\expT_2)}$.
  From (3), (4) and Lemma~\ref{le:downcast-red-equiv}  we obtain that
  $\redLRConfk{k}{t}{\vbFG(e).(t)}{\expT_1 \ \vbTL(\expT_2)}$.

  We have that $\vbTL(\expT_1) = \expT_1$ and thus we are done for this case.
  \qed
\end{proof}

\subsection{Lemma~\ref{le:method-red-rel-equiv}}

  \begin{proof}
    Based on rules \Rule{red-rel-decls} and \Rule{red-rel-method},
    for
    $$\FUNC\ (x \ t_S) \ m (\foreach{x_i \ t_i}{n}) \ t \ \{ \RETURN\ e \} \in \foreachN{D}$$
    we have to show that

   \begin{mathpar}
     \forall k' \leq k,
                v', \uT', \foreach{v_i}{n}, \foreach{\uT_i}{n}.
                    (\redLRConfk{k'}{t_S}{v'}{\uT'}
                    \wedge (\forall i \in [n]. \redLRConfk{k'}{t_i}{v_i}{\uT_i}))
         \\ \implies (1) \ \redLRConfk{k'}
               {t}
               {\Angle{\subst{x}{v'},\foreach{\subst{x_i}{v_i}}{n}} e}
               {(\mT{x}{t_S} \ \uT') \ (\foreach{\uT_i}{n})}
    \end{mathpar}

   We verify the result by induction on $k$.

   \noindent
       {\bf Case $k=1$:} We must perform several reductions on $(\mT{x}{t_S} \ \uT') \ (\foreach{\uT_i}{n})$
       to obtain a value. Due to $k=1$ the premise of rule \Rule{red-rel-exp} holds vacuously.
       Therefore, we can immediately establish (1).

 \noindent
     {\bf Case $k \implies k+1$:}
     Suppose $k' \leq k+1$
     and (2) $\redLRConfk{k'}{t_S}{v'}{\uT'}$
     and (3) $\redLRConfk{k'}{t_i}{v_i}{\uT_i}$
     for some $v'$, $\uT'$, $v_i$, $\uT_i$ for $i \in [n]$.

     Suppose $\Angle{\subst{x}{v'},\foreach{\subst{x_i}{v_i}}{n}} e$
     and $(\mT{x}{t_S} \ \uT') \ (\foreach{\uT_i}{n})$ are reducible.
     Otherwise, the result holds immediately.

     We have to show that for
     (4) $\reduceFGk{k_1}{\foreachN{D}}{\Angle{\subst{x}{v'},\foreach{\subst{x_i}{v_i}}{n}} e}{v''}$
     and
     (5) $\reduceTLk{k_2}{\vbMethTL}{(\mT{x}{t_S} \ \uT') \ (\foreach{\uT_i}{n})}{\uT''}$
     and $k+1-k_1 - k_2 > 0$
     we have that
     (6) $\redLRConfk{k + 1 - k_1 - k_2}{t}{v''}{\uT''}$.

     From (5) we can conclude that
     (7) $\vbMethTL \turnsTL (\mT{x}{t_S} \ \uT') \ (\foreach{\uT_i}{n})
     \reduceSym^{1}
     \lambda X . \lambda (\foreach{X_i }{n}). \expT
     \reduceSym^{2} \Angle{\subst{\xT}{\uT'},\foreach{\subst{\xT_i}{\uT_i}}{n}} \expT
     \reduceSym^{k_2'} \uT''$ where (8) $k_2 = k_2' + 3$.

     By induction we have that (9) $\redLRk{k}{}{}{\foreachN{D}}{\vbMethTL}$.

     From (2) and (3) and the Monotonicity Lemma~\ref{le:monotonicity}
     we find that
     (10) $\redLRConfk{k''}{t_S}{v'}{\uT'}$
     and (11) $\redLRConfk{k''}{t_i}{v_i}{\uT_i}$
     where $k'' \leq k$ for $i \in [n]$.

     By making use of (9), (10) and (11) we apply Lemma~\ref{le:exp-red-equivalent}
     on
     $$\tdMethTrans{\foreachN{D}}
                        {\FUNC\ (x \ t_S) \ m (\foreach{x_i \ t_i}{n}) \ t \ \{ \RETURN\ e \}}
                        {\lambda \xT . \lambda (\foreach{\xT_i }{n}). \expT}$$
     and thus obtain that
     (12) $\redLRConfk{k}{t}{\Angle{\subst{x}{v'},\foreach{\subst{x_i}{v_i}}{n}} e}{\Angle{\subst{\xT}{\uT'},\foreach{\subst{\xT_i}{\uT_i}}{n}} \expT}$.

     From (12), (4) and (7) via rule \Rule{red-rel-exp} we conclude that
     (11) $\redLRConfk{k-k_1-k_2'}{t}{v''}{\uT''}$.

     From (8) we conclude that
     (12) $k + 1 - k_1 - k_2 = k - k_1 - k_2' - 2$.

     From (11), (12) and the Monotonicity Lemma~\ref{le:monotonicity} we conclude that
     $\redLRConfk{k-k_1-k_2'-2}{t}{v''}{\uT''}$ and we are done.
     \qed
  \end{proof}

\end{document}